\documentclass[journal,comsoc]{IEEEtran}
\usepackage{graphicx}
\usepackage{amssymb}
\usepackage{amsmath}
\usepackage{cite}
\usepackage{subfigure}
\usepackage{mathrsfs}
\usepackage[displaymath,mathlines]{lineno}
\usepackage{color}
\usepackage{tabulary}
\usepackage{multirow}
\usepackage{algpseudocode}
\usepackage{algorithm,algpseudocode}
\usepackage{algorithmicx}
\usepackage{pbox}
\usepackage{multicol}
\usepackage{lipsum}
\usepackage{amsthm}
\usepackage{relsize}
\usepackage{lipsum}
\usepackage{epstopdf}
\usepackage{mathtools}% http://ctan.org/pkg/mathtools
\usepackage{calc}% http://ctan.org/pkg/calc
\usepackage{makecell}

\makeatletter
\newcommand{\doublewidetilde}[1]{{%
		\mathpalette\double@widetilde{#1}}}
\newcommand{\double@widetilde}[2]{%
		\sbox\z@{$\m@th#1\widetilde{#2}$}%
		\ht\z@=.5\ht\z@
		\widetilde{\box\z@}}
\makeatother

\newtheorem{theorem}{Theorem}
\newtheorem{lemma}{Lemma}
\newtheorem{corollary}{Corollary}

\begin{document}
%\font\myfont=cmr4 at 12pt
\title{\huge Joint Power Allocation and User Scheduling in Integrated Satellite-Terrestrial Cell-Free Massive MIMO IoT Systems}

%\author{Trinh Van Chien
\author{Trinh~Van~Chien, \textit{Member}, \textit{IEEE}, Ha~An~Le, Ta Hai Tung, Hien~Quoc~Ngo, \textit{Senior Member}, \textit{IEEE}, and \\ Symeon~Chatzinotas, \textit{Fellow}, \textit{IEEE}  \vspace*{-1cm}
\thanks{Trinh Van Chien and Ta Hai Tung are with the School of Information and Communication Technology (SoICT), Hanoi University of Science and Technology (HUST), 100000 Hanoi, Vietnam (email: chientv@soict.hust.edu.vn, tungth@soict.hust.edu.vn). Ha An Le is with the Department of Electrical and Computer Engineering, Seoul National University, Korea (email: 25251225@snu.ac.kr).  Hien Quoc Ngo is with the School of Electronics, Electrical Engineering and Computer Science, Queen's University Belfast, Belfast BT7 1NN, United Kingdom (email: hien.ngo@qub.ac.uk). Symeon Chatzinotas is with the Interdisciplinary Centre for Security, Reliability and Trust (SnT), University of Luxembourg, L-1855 Luxembourg, Luxembourg (email: symeon.chatzinotas@uni.lu).} 
%}
}

\maketitle

\begin{abstract}
Both space and ground communications have been proven effective solutions under different perspectives in Internet of Things (IoT) networks. This paper investigates multiple-access scenarios, where plenty of IoT users are cooperatively served by a satellite in space and access points (APs) on the ground. Available users in each coherence interval are split into scheduled and unscheduled subsets to optimize limited radio resources. We compute the uplink ergodic throughput of each scheduled user under imperfect channel state information (CSI) and non-orthogonal pilot signals. As maximum-radio combining is deployed locally at the ground gateway and the APs, the uplink ergodic throughput is obtained in a closed-form expression. The analytical results explicitly unveil the effects of channel conditions and pilot contamination on each scheduled user. By maximizing the sum throughput, the system can simultaneously determine scheduled users and perform power allocation based on either a model-based approach with alternating optimization or a learning-based approach with the graph neural network. Numerical results manifest that integrated satellite-terrestrial cell-free massive multiple-input multiple-output systems can significantly improve the sum ergodic throughput over coherence intervals. The integrated systems can schedule the vast majority of users; some might be out of service due to the limited power budget. 
\end{abstract}

% no keywords

% \fi
\IEEEpeerreviewmaketitle

\vspace*{-0.25cm}
\section{Introduction}
Wireless communications, especially Internet of Things (IoT) systems, have undergone a remarkable transformation, and notably, terrestrial networks are the dominant modes that provide enhanced communication speeds and quality of service (QoS)  \cite{guo2021enabling}. With mobile phones or other intelligent users, broadband services with low latency can be accessed within the ground access points (APs) range  through the utilization of joint coherent transmission techniques defined in cell-free massive mutiple-input multiple-output (MIMO) communications \cite{ngo2017cell}. Looking ahead to the future, the sixth-generation (6G) wireless network is anticipated to deal with  an unprecedented surge in device demand. The coverage requirements of 6G will be crucial to support widely distributed devices across vast areas, including humans, machines, and various interconnected objects. According to some estimates, the number of devices will approach over 24 billion by 2030 \cite{Li2018-IOT}. However, due to some limitations, including geographical locations and operation costs, it will be very challenging to guarantee coverage solely using terrestrial cellular networks. Specifically, the terrestrial networks are often deployed in areas with high-density populations for economic benefits \cite{Kuang2017-TScomm}. The vast airspace and sea areas are not fully covered by traditionally mobile networks due to geographical topology. Based on the Global System for Mobile Communications Assembly (GMSA) report, over $40\%$ of the world's surface lacks network coverage, leaving a significant portion without access to communication networks. Additionally, approximately $4.6$ billion internet users eagerly await an improved network with higher speed and reduced latency \cite{GSMA2019}.

To address this issue, satellite communication networks offer an immediate solution to coverage problems by providing extensive coverage capabilities \cite{Su2019-LEO}. Satellite communication networks could complement terrestrial networks and provide global coverage with ubiquitous connectivity. Current terrestrial cellular networks have provided a promising solution to handle the rapid growth of massive connectivity for IoT networks in spheres of life \cite{da2014internet}. 
LEO satellites orbit the Earth in a circular (or elliptical) pattern between $250$ and $2000$~km above the surface \cite{sciddurlo2021looking} offer distinctive merits to connect terrestrial devices on the ground, which can communicate objects with very limited or even no access to traditional terrestrial networks.  The integration of satellite technology into ground networks presents remarkable opportunities for the advancement of future wireless communications, i.e., beyond 5G and toward 6G radio communications \cite{van2022space,zhu2021integrated}. The potential integration architectures of the two networks have recently been discussed in various publications such as \cite{zhao2022interlink,jiang2021qoe, zhang2021stochastic} and references therein. By integrating terrestrial relays into satellite networks, the  satellite-terrestrial architecture can, in comparison to a traditional single network, first aid in improving communication dependability \cite{nguyen2022security}. Moreover, on the basis of satellite backhaul transmission, the integrated satellite-terrestrial design assists in extending network coverage effectively \cite{al2021modeling}. Besides, time and frequency sharing within these integrated network architectures help upgrade spectral and energy efficiency productively \cite{zhang2019spectrum}. By virtue of the network cooperation, the integrated satellite-terrestrial systems are capable of guaranteeing  seamless service connectivity and providing improved transmission \cite{wang2022mega}. Nonetheless, most of the above-mentioned related works consider perfect channel state information. In addition, their resource allocations  rely on slow-fading channel models, which may be impossible to deploy in practice under high mobility. There is still  room for analyzing network performance and allocating radio resources applicable for an extended period with lower computational complexity by only exploiting channel statistics. 

Machine learning (ML) in general, especially deep learning in particular, has appeared as a promising technology to handle numerous complicated problems in radio communication systems and IoT networks, including channel estimation \cite{Neumann2018MLCE}, radio resource allocation \cite{Sun2018DNNresource,chun2023data}, and signal decoding\cite{OShea2017wirelessAE}. There exist two main approaches in the literature to designing machine learning-based schemes for wireless systems. The first approach is the data-driven approach that exploits neural networks to learn the optimal mapping between the input and the output of the objective. For example, in \cite{Liang2020DNNpower}, a fully connected multi-layer perceptrons (MLPs) is proposed to learn and predict the mapping between instantaneous channels and the optimal power allocation for a $K$-user single-antenna inference system. The second approach is model-driven, which exploits neural networks to replace ineffective policies in classical algorithms \cite{Hengtao2019ModelDriven}. While both approaches can achieve a near-optimal solution compared to conventional methods with much faster execution time in a small-scale network, their performance degrades significantly in large-scale systems with multiple dimensions. For instance, it was demonstrated in \cite{Yifan2021BFneural} that the proposed CNN model for the beamforming problems can achieve performance close to the conventional approach in a two-user network. Still, an 18\% gap occurs for a 10-user network. Furthermore,  these ML-based methods generalize poorly with the significantly dropped performance when the system settings in the test dataset differ from the training  \cite{Yifei2020resource}. These disadvantages prevent machine learning models from being applicable in real-life communication systems, where the system setup often changes dramatically. 

To improve the scalability and generalization of ML-based methods, a promising approach is to embrace the features of the wireless topology into neural network architectures. Graph neural network (GNN) is a well-known approach that can explore the graph topology of radio systems to obtain a comparable performance and remarkable scalability and generalization in very large-scale systems \cite{Yifei2021GNNresource,Junbeom2023GNNBeamforming,Eisen2020GNNresource} . Specifically, it was shown in \cite{Yifei2021GNNresource} that a GNN model trained with $50$ users could achieve similar performance in a much large network with $1000$ users in a resource management problem. Moreover, it is proven in \cite{Yifei2022GNNWireless} that GNN models can obtain a comparable performance in large-scale systems. In contrast, the prediction performance of MLP degrades severely as the number of system parameters enlarges. However, to the best of the authors' knowledge, no related works are designing a GNN for integrated satellite-terrestrial cell-free massive MIMO systems to learn and predict the spectral efficiency  with heterogeneous users from space and ground.

This paper considers an integrated satellite-terrestrial cell-free massive MIMO IoT system where a satellite and multiple APs jointly serve many terrestrial users under practical communication conditions. To the end, the main contributions of this paper are summarized as follows:
\begin{itemize}
	\item We investigate a category of cooperative networks with a presence of a LEO satellite in which users can be either in active or inactive mode according to channel conditions and finite radio resources. The instantaneous CSI is  estimated at the gateway and APs in the uplink pilot training by exploiting the minimum mean square error (MMSE) estimation scheme. To keep a generic framework, we assume an arbitrary pilot reuse pattern. 
	\item From the channel estimates and estimation errors, we derive an uplink ergodic  throughput of each active user, which can be applied to any detection method and channel model. This throughput is then computed in closed form for the maximum-ratio combining (MRC) and spatially correlated Raleigh fading channel model.
	\item By considering the transmit data power coefficients as the variables, we investigate an optimization problem, which maximizes the active users' total ergodic throughput subject to the limited power resource constraints. Despite the inherent non-convexity, this sum ergodic throughput optimization problem allows analyzing system performance with resource management in satellite-terrestrial systems and obtaining the solution to both the transmit power control and user scheduling under multiple access.  
	\item For the model-based approach, we come up with an iterative algorithm that enables to yield a stationary solution to the considered problem in polynomial time by exploiting the alternating optimization (AO) method. In each iteration, the closed-form expression of an optimization variable is derived by investigating the first-order derivative of the Lagrangian function and conditioning the remaining optimization variables.  
	\item For the learning-based approach, we construct a heterogeneous GNN that, distinguished from previous works, only exploits statistical information to optimize the transmit power to every user and schedule all users in the coverage area. Unsupervised learning is
	exploited to train the GNN with only statistical channel state information from  APs and the satellite. 
	   
	\item Numerical results qualify the correctness of our analytical framework for the uplink ergodic  throughput. The proposed optimization algorithm effectively allocates the power budget and schedules the users. The low running time and scalability of the learning-based approach are testified under statistical information deployment.
\end{itemize}
The rest of this paper is organized as follows: Section~\ref{Sec:SysModel} presents in detail the considered system model and the channel estimation procedure from the uplink pilot training phase. After that, the uplink data transmission and the analysis of ergodic throughput are presented in Section~\ref{Sec:UplinkPerfAna} with a closed-form expression obtained as the satellite and APs exploit the MRC technique. The joint sum throughput maximization and user scheduling with respect to the limited power budget at each terrestrial user is formulated and solved in Section~\ref{Sec:SumRate} by using a model-based approach. By using GNN and unsupervised learning, we present a learning-based approach to handle the joint power and user scheduling for maximizing the sum ergodic throughput. Section~\ref{Sec:NumericalResults} provides numerical results extensively, while main conclusions are finally given in Section~\ref{Sec:Conclusion}. 

\textit{Notation}: Lower and upper bold letters are utilized to express vectors and matrices. Meanwhile, the superscript $(\cdot)^H$ and $(\cdot)^T$ are  Hermitian and regular transpose, respectively. $\mathrm{tr}(\mathbf{A})$ denotes the trace of square matrix $\mathbf{A}$, whilst an identity matrix of size $N \times N$ is denoted by $\mathbf{I}_N$.  The expectation of a random variable is $\mathbb{E}\{\cdot\}$ and $\mathcal{CN}(\cdot, \cdot)$ stands for the circularly symmetric Gaussian distribution. Meanwhile, let us denote $\mod(\cdot, \cdot)$  as the modulus operation, and $\lfloor \cdot \rfloor$ as the floor function.
\begin{figure}[t]
	\centering
	\includegraphics[trim=1.5cm 0.0cm 0.0cm 0.0cm, clip=true, width=3.0in]{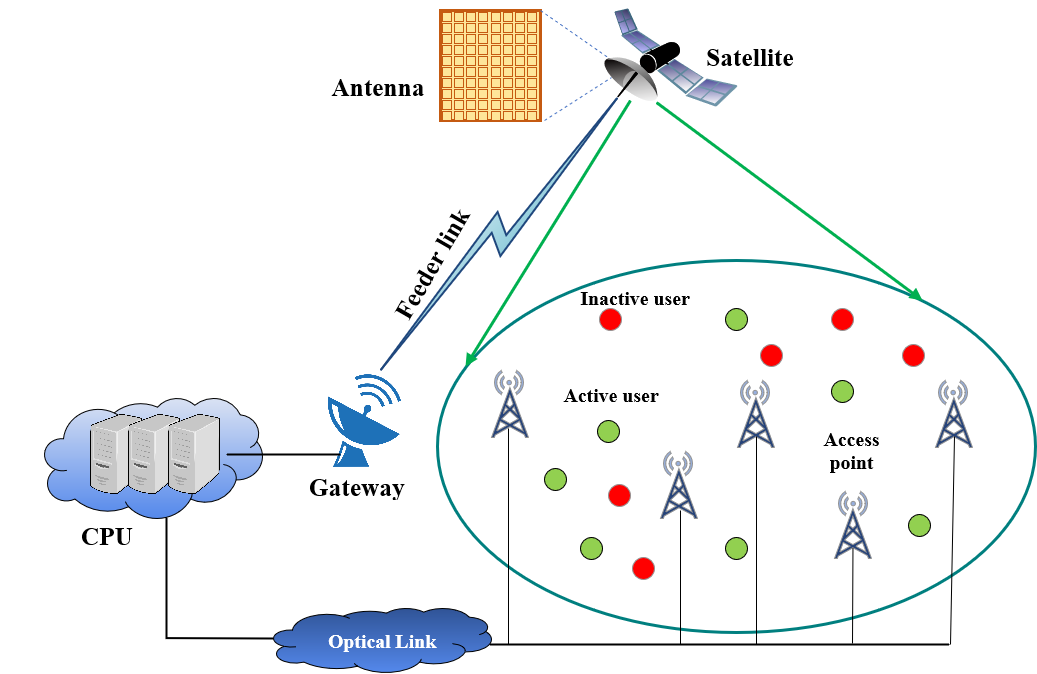} \vspace*{-0.2cm}
	\caption{The considered satellite-terrestrial cooperative IoT network where $M$ APs and a LEO satellite jointly serve the $K$ users with both active (green color) and inactive  (red color) users.}
	\label{FigSysModel}
	\vspace*{-0.5cm}
\end{figure}
\vspace*{-0.25cm}
\section{System Model and Uplink Pilot Training} \label{Sec:SysModel}
\vspace*{-0.25cm}
We consider a cooperative wireless network that includes $M$ APs and $K$ available IoT users, all having a single antenna. Let us denote the universal set $\mathcal{K} = \{1, \ldots, K\}$ comprising all available IoT users, and $\mathcal{M} = \{1, \ldots, m\}$ the set of all APs.  To enhance the system,  a LEO satellite with $N$ antennas is exploited. The network consideration in this paper is illustrated in Fig.~\ref{FigSysModel}. The considered system exploits orthogonal frequency division multiplexing (OFDM), so the block fading channel model is utilized in each OFDM subcarrier.  The instantaneous channels at the APs and satellite are estimated by the uplink pilot training phase. More precisely, in each coherence interval with $\tau_c$ symbols, where the propagation channels of the space and ground communications are quasi-static and frequency flat. Among them, $\tau_p$ symbols are used for the pilot training, and the remaining comprising $\tau_c - \tau_p$ symbols are dedicated to data transmission in the uplink. Due to massive connectivity, a subset of users, defined adequately by solving resource allocation problems, may be dropped from the service during the data transmission. Thus, we denote $\mathcal{Q}$ the set of active users with $\mathcal{Q} \subseteq \mathcal{K}$ and the complement of $\mathcal{Q}$, which is  $\bar{\mathcal{Q}} =  \mathcal{K}\setminus \mathcal{Q}$ comprising inactive users that are out of the service. 

By assuming a rich scattering environment where many scatterers surround the users,  the channel link  between AP~$m$, $\forall m,$ and user~$k$, $\forall k,$ i.e., $g_{mk} \in \mathbb{C}$, follows a Rayleigh fading model, which is $g_{mk}  \sim \mathcal{CN}(0, \beta_{mk})$, where $\beta_{mk}$ denotes the large-scale fading including, for example, both path loss caused by propagation distance and shadow fading such as large obstacles and buildings. Meanwhile, the space channel between the satellite and user~$k$, denoted by $\mathbf{g}_k \in \mathbb{C}^N$, follows a Rician distribution, i.e., $\mathbf{g}_k \sim \mathcal{CN}(\bar{\mathbf{g}}_k, \mathbf{R}_k)$,
in which $\bar{\mathbf{g}}_k \in \mathbb{C}^N$ stands for the LoS components
and $\mathbf{R}_k \in \mathbb{C}^{N \times N}$ indicates the spatial correlation matrix. The channel models are of practical interest and closer to reality, where both the propagation effects and the spatial correlation from the antenna structure are considered. 

\vspace*{-0.25cm}
\subsection{Uplink Pilot Training}
In the considered system, the propagation channels are estimated in the uplink training phase by letting each user transmit a pilot signal, including $\tau_p$ symbols dedicated in every coherence block. We assume that all the $K$ users are involved in the pilot training phase for the network to know the channel information. The same set of $\tau_p$ orthonormal pilot signals are reused across the users, say $\{ \pmb{\phi}_1, \ldots, \pmb{\phi}_{\tau_p}\}$, in which the pilot signal $\pmb{\phi}_k \in \mathbb{C}^{\tau_p}$ is designated to user~$k$. Let us denote $\mathcal{P}_k \subseteq \mathcal{K}$ the subset of user indices that share the same pilot signal as user~$k$ and create the following pilot reuse pattern
\begin{equation} \label{eq:PilotPattern}
\pmb{\phi}_k^H \pmb{\phi}_{k'} = \begin{cases}
	1, & \mbox{if } k' \in \mathcal{P}_k,\\
	0, & \mbox{otherwise}.
\end{cases}
\end{equation}
For the ground link, AP~$m$ receives the training signal, denoted by $\mathbf{y}_{pm} \in \mathbb{C}^{\tau_p}$, is superimposed of all the pilot signals sent over the terrestrial links as
\begin{equation} \label{eq:ypm}
	\mathbf{y}_{pm} = \sum\nolimits_{k = 1}^K \sqrt{p\tau_p} g_{mk} \pmb{\phi}_k^H + \mathbf{w}_{pm}^H,
\end{equation}
where $p$ is the transmit power, which users can grant to every pilot symbol in $\pmb{\phi}_k, \forall k,$ and $\mathbf{w}_{pm} \sim \mathcal{CN}(\mathbf{0}, \sigma_a^2 \mathbf{I}_{\tau_p})$ is additive noise at AP~$m$ with standard derivation $\sigma_a$~[dB] and zero mean. After that, AP~$m$ estimates the desired channel from user~$k$ by projecting the received training signal in \eqref{eq:ypm} onto $\pmb{\phi}_k$ as
\begin{equation} \label{eq:ProjChannelypmk}
y_{pmk} = \sqrt{p\tau_p} g_{mk} + \sum\nolimits_{k' \in \mathcal{P}_k \setminus \{k\}} \sqrt{p\tau_p} g_{mk'} + \mathbf{w}_{pm}^H \pmb{\phi}_k.
\end{equation}
For space communications, the received training signal at the gateway of satellite, $\mathbf{Y}_p \in \mathbb{C}^{N \times \tau_p}$, can be formulated in a similar manner as
\begin{equation}
\mathbf{Y}_p = \sum\nolimits_{k=1}^K \sqrt{p\tau_p} \mathbf{g}_{k} \pmb{\phi}_k^H + \mathbf{W}_{p},
\end{equation}
where $\mathbf{W}_{p} \in \mathbb{C}^{N \times \tau_p}$ is additive noise whose elements distributed as $\mathcal{CN}(0, \sigma^2)$. The desired propagation channel from user~$k$ to the satellite is gathered at the gateway by projecting $\mathbf{Y}_p$ onto $\pmb{\phi}_k$ as
\begin{equation} \label{eq:ProjSigypk}
\mathbf{y}_{pk} = \mathbf{Y}_p \pmb{\phi}_k =  \sqrt{p\tau_p} \mathbf{g}_{k}  +  \sum\nolimits_{k' \in \mathcal{P}_k \setminus \{k\} } \sqrt{p\tau_p} \mathbf{g}_{k'} + \tilde{\mathbf{w}}_{pk},
\end{equation}
where $\tilde{\mathbf{w}}_{pk} = \mathbf{W}_{p} \pmb{\phi}_k^H$ is additive noise at the satellite section, which is weighted by the pilot signal $\pmb{\phi}_k$ and distributed as $\tilde{\mathbf{w}}_{pk} \sim \mathcal{CN}(\mathbf{0}, \sigma_s^2 \mathbf{I}_N)$ with zero mean and standard deviation $\sigma_s$~[dB]. The network will deploy the MMSE estimation to obtain the channel estimates along with the estimation errors as in Lemma~\ref{Lemma:Est}.
\begin{lemma} \label{Lemma:Est}
By exploiting the MMSE estimation locally at each AP, the estimate of the channel $g_{mk}$ between AP~$m$ and user~$k$ can be formulated based on \eqref{eq:ProjChannelypmk} as
\begin{equation} \label{eq:ChanEstgmk}
\hat{g}_{mk} = \mathbb{E}\{ g_{mk} | y_{pmk}\} = c_{mk} y_{pmk},
\end{equation}
where $c_{mk} = \mathbb{E} \{  y_{pmk}^\ast g_{mk} \} / \mathbb{E} \{ | y_{pmk} |^2 \} $ is computed as
\begin{equation}
c_{mk} = \frac{\sqrt{p\tau_p} \beta_{mk} }{ \sum_{k' \in \mathcal{P}_k} p \tau_p  \beta_{mk'} + \sigma_a^2}.
\end{equation}
From \eqref{eq:ChanEstgmk}, we observe that the channel estimate $\hat{g}_{mk}$ is distributed as $\hat{g}_{mk} \sim \mathcal{CN}(0, \gamma_{mk})$, where
\begin{equation}
\gamma_{mk} = \mathbb{E} \{ |\hat{g}_{mk}|^2\} =  \frac{ p\tau_p \beta_{mk}^2 }{ \sum_{k' \in \mathcal{P}_k} p \tau_p  \beta_{mk'} + \sigma_s^2}.
\end{equation}
In addition, let us define the channel estimation error $e_{mk} = g_{mk} - \hat{g}_{mk}$, then $e_{mk} \sim \mathcal{CN}(0, \beta_{mk} - \gamma_{mk} )$. Note that $\hat{g}_{mk}$ and $e_{mk}$, $\forall m,k,$ are independent random variables.

In a similar manner, the channel estimate $\hat{\mathbf{g}}_k$ of the propagation channel $\mathbf{g}_k$ between the satellite and user~$k$ can be formulated based on \eqref{eq:ProjSigypk} as
\begin{equation} \label{eq:ChannelEstgk}
\hat{\mathbf{g}}_k = \bar{\mathbf{g}}_k +  \sqrt{p\tau_p} \mathbf{R}_k \pmb{\Phi}_k ( \mathbf{y}_{pk} -  \bar{\mathbf{y}}_{pk} ),
\end{equation}
where $\bar{\mathbf{y}}_{pk} = \sum_{k' \in \mathcal{P}_k} p \tau_p  \bar{\mathbf{g}}_k$ and $\pmb{\Phi}_k = \big( \sum_{k' \in \mathcal{P}_k} p \tau_p \mathbf{R}_{k'} + \sigma_s^2 \mathbf{I}_N \big)^{-1}$. Then, the channel estimate $\hat{\mathbf{g}}_k$ and the channel estimation error $\mathbf{e}_k = \mathbf{g}_k - \hat{\mathbf{g}}_k$ are respectively distributed as 
\begin{align}
&\hat{\mathbf{g}}_k \sim \mathcal{CN}(\bar{\mathbf{g}}_k,  p \tau_p \mathbf{R}_k \pmb{\Phi}_k \mathbf{R}_k),\\
& \mathbf{e}_k \sim \mathcal{CN}(\mathbf{0},  \mathbf{R}_k - p \tau_p \mathbf{R}_k \pmb{\Phi}_k \mathbf{R}_k).  \label{eq:Ck}
\end{align}
 We observe that $\hat{\mathbf{g}}_k$ and $\mathbf{e}_k$, $\forall k$, are independent random variables. 
\end{lemma}
\begin{proof}
The proof adopts the standard MMSE estimation \cite{Kay1993a} to the integrated system model and notations.
\end{proof}
Lemma~\ref{Lemma:Est} provides the concrete expressions of the channel estimates that are utilized for designing the combining coefficients to detect the desired signals under an arbitrary pilot reuse pattern. For the terrestrial links, the channel estimates between two users~$k$ and $k'$ sharing the same pilot signal unveil the following relationship
\begin{equation} \label{eq:Relation}
\hat{g}_{mk}/c_{mk} = \hat{g}_{mk'}/c_{mk'}, ~ \text{and hence,} ~ \gamma_{mk}/c_{mk}^2 = \gamma_{mk'}/c_{mk'}^2,
\end{equation}
which indicates that the network cannot differentiate the channel estimates of these two users since one is a scaling-up version of the other. This behavior is not explicitly observed in the space links under the spatial correlation at the satellite. However, the active users in the set $\mathcal{K}$ share the same matrix $\pmb{\Phi}_k$ and they are distinguished by the spatial covariance matrices. One way to mitigate the channel estimation errors is allowing each user to occupy its own pilot signal that leads to $\pmb{\Phi}_k \big(  p \tau_p \mathbf{R}_{k} + \sigma_s^2 \mathbf{I}_N \big)^{-1}$. Nevertheless, the orthogonal pilot assignment is impossible under a short coherence time with many available users. 
\vspace*{-0.25cm}
\section{Analysis of Ergodic Throughput and Uplink Data Transmission} \label{Sec:UplinkPerfAna}
This section provides the analysis of the throughput in the uplink data transmission under imperfect channel state information. A closed-form expression of the uplink throughput with MRC is then derived.  
\vspace*{-0.25cm}
\subsection{Uplink Data Transmission}
 Active users in set $\mathcal{Q}$ are allowed to access the network in the uplink data transmission such that a particular utility metric can be optimized with a finite radio resource. From this assumption, the received signal at the satellite, i.e., $\mathbf{y} \in \mathbb{C}^N$, and that of  AP~$m$, i.e., $y_m \in \mathbb{C}$ are defined as
\begin{align} 
\mathbf{y} &= \sum\nolimits_{k \in \mathcal{Q}} \sqrt{\rho_k} \mathbf{g}_k s_k  + \mathbf{w} \mbox{ and } y_m =  \sum\nolimits_{k \in \mathcal{Q}} \sqrt{\rho_k} g_{mk} s_k  + w_m, \label{eq:ReceiveSigAP}
\end{align}
where $\mathbf{w} \sim \mathcal{CN}(\mathbf{0}, \sigma_s^2 \mathbf{I}_N)$   is additive noise at the satellite system and $w_m \sim \mathcal{CN}(0,\sigma_a^2)$ is that of AP~$m$. From the received signals in \eqref{eq:ReceiveSigAP}, we will decode the desired signal $s_k$ sent by user~$k$, $\forall k$, by an advanced process with the two-layer decoding technique. The desired signal  transmitted from user~$k$ is decoded independently at the gateway, $tilde{s}_k = \mathbf{u}_k^H \mathbf{y} $,  and at each AP, $\tilde{s}_{mk} = u_{mk}^\ast y_m$ as follows
\begin{align}
& \tilde{s}_k = \sqrt{\rho_k} \mathbf{u}_k^H \mathbf{g}_k  s_k  + \sum\nolimits_{k' \in \mathcal{Q}, k'\neq k} \sqrt{\rho_{k'}} \mathbf{u}_k^H \mathbf{g}_{k'} s_{k'} + \mathbf{u}_k^H \mathbf{w},  \\
& \tilde{s}_{mk} =   \sqrt{\rho_k} u_{mk}^\ast g_{mk} s_k + \sum\nolimits_{k' \in \mathcal{Q}, k'\neq k} \sqrt{\rho_{k'}}  u_{mk}^\ast g_{mk'} s_{k'}  + w_m,
\end{align}
where $\mathbf{u}_k \in \mathbb{C}^{N}$ is the combining vector exploited to decode the desired signal sent from the satellite over the space link and  the combining coefficient exploited by AP~m is denoted by $u_{mk} \in \mathbb{C}$. All the decoded signals for user~$k$, i.e., $\tilde{s}_k$ and $\tilde{s}_{mk}, \forall m,$ will be combined at the central processing unit (CPU) as follows
\begin{multline} \label{eq:2ndsk}
\hat{s}_k  = \tilde{s}_k + \sum\nolimits_{m=1}^M   \tilde{s}_{mk} = \sqrt{\rho_k} \left( \mathbf{u}_k^H \mathbf{g}_k + \sum\nolimits_{m=1}^M u_{mk}^\ast g_{mk} \right) s_k  \\  +  \sum\nolimits_{k'\in \mathcal{Q}, k'\neq k} \sqrt{\rho_{k'}} \left(\mathbf{u}_k^H \mathbf{g}_{k'}  +  \sum\nolimits_{m=1}^M u_{mk}^\ast  g_{mk'} \right)s_{k'} 
 \\ +  \mathbf{u}_k^H \mathbf{w} + \sum\nolimits_{m=1}^M  u_{mk}^\ast w_m.
\end{multline}
The first term in  \eqref{eq:2ndsk} represents the received signal sent by the desired user~$k$, which inherits the diversity gain from the communication channels from both satellite and APs. The second term represents coherent and noncoherent  interference aggregated from all the other users gathered at the APs and satellite. The remaining terms are additive noise. We stress that the decoded signal in \eqref{eq:2ndsk} is a generalization of the related works in the literature, which can be reduced to either  ground or space communication by removing the corresponding parts.
\vspace*{-0.25cm}
\subsection{Uplink Ergodic Throughput}
\vspace*{-0.125cm}
When the number of antennas at the satellite and APs grows large sufficiently, the network can tackle the channel gain of the signal $s_k$ from the desired user~$k$ in \eqref{eq:2ndsk} as a constant. Hence, one can compute the uplink ergodic throughput of user~$k$ effectively. We now  introduce a new notation
\begin{equation} \label{eq:zkkprime}
z_{kk'} = \mathbf{u}_k^H \mathbf{g}_{k'}  +  \sum\nolimits_{m=1}^M  u_{mk}^\ast  g_{mk'},
\end{equation}
which contains the overall received channel information. For $k=k'$, the notation $z_{kk}$ denotes the desired channel at user~$k$. Otherwise, $z_{kk'}$ represents an interfering channel that degrades the received signal strength at the receiver. By exploiting \eqref{eq:zkkprime}, the decoded signal in \eqref{eq:2ndsk} is equivalent to as 
\begin{equation} \label{eq:Decodesig}
	\begin{split}
		 &\hat{s}_k  = \sqrt{\rho_k} z_{kk} s_k +   \sum\nolimits_{k'\in \mathcal{Q}, k'\neq k} \sqrt{\rho_{k'}} z_{kk'} s_{k'} + \mathbf{u}_k^H \mathbf{w} + \\
    &\sum\nolimits_{m=1}^M u_{mk}^\ast w_m = \sqrt{\rho_k} \mathbb{E} \{ z_{kk} \} s_k +  \sqrt{\rho_k} \left( z_{kk}  -  \mathbb{E} \{ z_{kk} \}  \right)s_k \\
    & +   \sum_{k'\in \mathcal{Q}, k'\neq k} \sqrt{\rho_{k'}} z_{kk'} s_{k'} +  \mathbf{u}_k^H \mathbf{w} + \sum_{m=1}^M  u_{mk}^\ast w_m,
	\end{split}
\end{equation}
where the first term in the last equality of \eqref{eq:Decodesig} involves the desired signal sent by user~$k$ equipped with a deterministic effective channel gain contributed from both the satellite and APs. The second term demonstrates the random fluctuation of the effective channel gain for a predetermined linear combining method, which is the so-called beamforming uncertainty. The remaining terms indicate mutual interference because of multiple users accessing the network simultaneously together with thermal noise. Owing to the so-called use-and-then-forget capacity bounding technique \cite{Chien2017a,massivemimobook}, the uplink ergodic throughput of user~$k$ is computed as
\begin{equation} \label{eq:Rkv1}
	R_k = B\left( 1 - \frac{\tau_p}{\tau_c} \right)  \log_2 ( 1 + \mathrm{SINR}_k ), \mbox{[Mb/s/Hz]},
\end{equation}
where $B$~[MHz] determines the operating bandwidth and $\mathrm{SINR}_k$ is the effective signal-to-interference-and-noise ratio (SINR) that is computed as in \eqref{eq:SINRk}. 
\begin{figure*}
\begin{equation} \label{eq:SINRk}
\mathrm{SINR}_k = \frac{\rho_k \big| \mathbb{E}\{ z_{kk}\} \big|^2 }{\sum_{k'\in \mathcal{Q}} \rho_{k'}  \mathbb{E}\{ |z_{kk'}|^2 \} - \rho_k \big| \mathbb{E}\{ z_{kk}\} \big|^2 +  \mathbb{E} \big\{ | \mathbf{u}_k^H \mathbf{w} |^2 \big\} + \sum_{m=1}^M  \mathbb{E} \big\{ | u_{mk}^\ast w_m |^2 \big\}  }
\end{equation}
%\hrule
\vspace{-0.5cm}
\end{figure*}
Note that the uplink ergodic throughput in \eqref{eq:Rkv1} can be applied for an arbitrary combining technique at satellite and APs. Though we can evaluate the expectation \eqref{eq:Rkv1} numerically, this approach requires various amounts of different realizations of random small-scale fading and shadow fading coefficients to attain the expectations numerically. Therefore, it may be burdensome for networks with low-cost hardware devices. By virtue of the fundamental massive MIMO properties, we now compute the closed-form solution to \eqref{eq:Rkv1} when the MRC combining method is deployed by the APs and the satellite as shown in the theorem below.
\begin{theorem} \label{Theorem:ClosedForm}
If the MRC method is exploited, the uplink ergodic throughput for user~$k$ in \eqref{eq:Rkv1} is computed in the closed form   as
\begin{equation} \label{eq:RateMRC}
R_k = B\left( 1 - \frac{\tau_p}{\tau_c} \right) \log_2( 1 + \mathrm{SINR}_k), \mbox{[Mbps]},
\end{equation}
where the effective SINR expression is 
\begin{equation} \label{eq:ClosedSINR}
 \mathrm{SINR}_k = \frac{\rho_k \left( \|\bar{\mathbf{g}}_k\|^2 +  p \tau_p \mathrm{tr}(\mathbf{R}_k \pmb{\Phi}_k \mathbf{R}_k)  +  \sum\nolimits_{m=1}^M \gamma_{mk} \right)^2}{  \mathsf{MI}_k + \mathsf{NO}_k},
\end{equation}
with the mutual interference, denoted by $\mathsf{MI}_k$,  and noise, denoted by $\mathsf{NO}_k$, given as follows
\begin{align}
&\mathsf{MI}_k =	\sum\nolimits_{k' \in \mathcal{P}_k \setminus \{k\}} \rho_{k'} \Big| \bar{\mathbf{g}}_{k}^H \bar{\mathbf{g}}_{k'} + p \tau_p   \mathrm{tr}(\mathbf{R}_{k'} \pmb{\Phi}_k \mathbf{R}_k) +   \sum\nolimits_{m=1}^M  \frac{c_{mk'}  }{c_{mk}  }\notag\\
& \times\gamma_{mk} \Big|^2 + \sum\nolimits_{k' \notin \mathcal{P}_k } \rho_{k'} |\bar{\mathbf{g}}_{k}^H  \bar{\mathbf{g}}_{k'} |^2  +  p \tau_p \sum\nolimits_{k' \in \mathcal{Q}} \rho_{k'}   \bar{\mathbf{g}}_{k'}^H \mathbf{R}_k \pmb{\Phi}_k \mathbf{R}_k  \bar{\mathbf{g}}_{k'}  \notag \\
&+ \sum\nolimits_{k'\in \mathcal{Q}} \rho_{k'} \bar{\mathbf{g}}_{k}^H \mathbf{R}_{k'} \bar{\mathbf{g}}_{k}   + p \tau_p \sum\nolimits_{k' \in \mathcal{Q} } \rho_{k'}   \mathrm{tr}( \mathbf{R}_{k'} \mathbf{R}_k \pmb{\Phi}_k \mathbf{R}_k ) \notag \\
& +  \sum\nolimits_{k' \in \mathcal{Q}} \sum\nolimits_{m=1}^M \rho_{k'}  \gamma_{mk} \beta_{mk'}, \label{eq:MIk}\\
& \mathsf{NO}_k =  \sigma_s^2 \|\bar{\mathbf{g}}_k\|^2 +  p \tau_p \sigma_s^2 \mathrm{tr}(\mathbf{R}_k \pmb{\Phi}_k \mathbf{R}_k ) + \sigma_a^2 \sum\nolimits_{m=1}^M \gamma_{mk}.
\end{align}
\end{theorem}
\begin{proof}
See Appendix~\ref{Appendix:ClosedForm} for the detailed proof.
\end{proof}
The uplink ergodic throughput obtained for user~$k$ in Theorem~\ref{Theorem:ClosedForm} depends on channel statistics only. Both the LoS components and the spatial correlation constructively contribute to upgrading the strength of desired signal $s_k$ sent from this user as pointed out in the numerator of \eqref{eq:ClosedSINR} thanks to the presence of the satellite. It also demonstrates the effectiveness of distributed APs over the coverage area generating the benefits of spatial diversity from the summation of $M$ terms. The satellite can be connected to a terrestrial network such that data processing at the CPU coherently results in a order of $M^2 N^2$ for the array gain. The denominator of \eqref{eq:ClosedSINR} shows the severity of mutual interference and additive noise degrading transmission performance. More specifically, the coherent interference can grow up with the quadratic order of the array gain from both the number of APs and satellite antennas, so the achievable throughput of each user is bounded from above at the limiting regime, i.e., $M, N \rightarrow \infty$. Without the presence of the satellite, the effective SINR expression in \eqref{eq:ClosedSINR} is simplified to the following expression, which is denoted by $\widehat{\mathrm{SINR}}_k$ as in \eqref{eq:SINRkGroundOnly}.
\begin{figure*}
\begin{equation} \label{eq:SINRkGroundOnly}
\widehat{\mathrm{SINR}}_k = \frac{\left|\sum\nolimits_{m=1}^M  \gamma_{mk} \right|^2}{	\sum\nolimits_{k' \in \mathcal{P}_k \setminus \{k\}} \rho_{k'} \left|  \sum\nolimits_{m=1}^M  \frac{c_{mk'} }{c_{mk}  }  \gamma_{mk} \right|^2 +  \sum\nolimits_{k' \in \mathcal{Q}} \sum\nolimits_{m=1}^M \rho_{k'}   \gamma_{mk} \beta_{mk'} + \sigma_a^2 \sum\nolimits_{m=1}^M \gamma_{mk}}
\end{equation}
\hrule
\vspace{-0.5cm}
\end{figure*}
which unveils that the received signal is only strengthened by  the macro-diversity gain from the distributed APs and the coherent gain of jointly processing the received data at the CPU. Specifically, the array gain is now only in the order of $M^2$. Nevertheless, the additive noise in \eqref{eq:SINRkGroundOnly} is less severe than that of in \eqref{eq:ClosedSINR} as the absence of the satellite. Both the SINR expressions in \eqref{eq:ClosedSINR} and \eqref{eq:SINRkGroundOnly} offers benefits in evaluating the system performance, e.g., spectral efficiency,  and reducing the computational complexity issue of updating the resource allocation algorithms thanks to the stability of the statistical channel information over many coherence intervals. Besides, if the space links are only available, we can reformulate the effective SINR expression in \eqref{eq:ClosedSINR} to as
\begin{equation} \label{eq:SINRSpaceOnly}
\widetilde{\mathrm{SINR}}_k = \frac{\rho_k \left( \|\bar{\mathbf{g}}_k\|^2 +  p \tau_p \mathrm{tr}(\mathbf{R}_k \pmb{\Phi}_k \mathbf{R}_k) \right)^2}{\widetilde{\mathsf{MI}}_k +  \widetilde{\mathsf{NO}}_k},
\end{equation}
where the following definitions hold for mutual interference and noise as
\begin{align}
	& \mathsf{MI}_k =	\sum\nolimits_{k' \in \mathcal{P}_k \setminus \{k\}} \rho_{k'} \left| \bar{\mathbf{g}}_{k}^H \bar{\mathbf{g}}_{k'} + p \tau_p   \mathrm{tr}(\mathbf{R}_{k'} \pmb{\Phi}_k \mathbf{R}_k)  \right|^2 + \notag \\
 & \sum\nolimits_{k' \notin \mathcal{P}_k } \rho_{k'} |\bar{\mathbf{g}}_{k}^H  \bar{\mathbf{g}}_{k'} |^2  +  p \tau_p \sum\nolimits_{k' \in \mathcal{Q}} \rho_{k'}   \bar{\mathbf{g}}_{k'}^H \mathbf{R}_k \pmb{\Phi}_k \mathbf{R}_k  \bar{\mathbf{g}}_{k'} +   \notag \\
	&  \sum\nolimits_{k'\in \mathcal{Q}} \rho_{k'} \bar{\mathbf{g}}_{k}^H \mathbf{R}_{k'} \bar{\mathbf{g}}_{k}   + p \tau_p \sum\nolimits_{k' \in \mathcal{Q} } \rho_{k'}   \mathrm{tr}( \mathbf{R}_{k'} \mathbf{R}_k \pmb{\Phi}_k \mathbf{R}_k ), \label{eq:MIkSpace} \\
	& \mathsf{NO}_k =  \sigma_s^2 \|\bar{\mathbf{g}}_k\|^2 +  p \tau_p \sigma_s^2 \mathrm{tr}(\mathbf{R}_k \pmb{\Phi}_k \mathbf{R}_k ),
\end{align}
which demonstrates the contributions of the satellite to the uplink throughput of user~$k$ only. Thanks to a strong channel gain with the LoS components, the received signal may be significantly enhanced as shown in the numerator of \eqref{eq:SINRSpaceOnly} with an array gain in the quadratic order of the number of satellite antennas. However, the LoS components also interfere each other in the same order as shown in \eqref{eq:MIkSpace}. The additive noise in \eqref{eq:SINRSpaceOnly} is less serve than in the joint satellite-space communication systems. 
\section{Sum Ergodic Throughput Optimization} \label{Sec:SumRate}
This section formulates and solves a sum ergodic throughput maximization problem constrained by the finite transmit power of the users. Furthermore,  active users are explicitly defined from the stationary solution of the optimized power coefficients.
\subsection{Problem Formulation}
In multiple access scenarios, the transmit power allocated to each user is a sophisticated function of the traffic load and the users' locations; however, the transmit power coefficients are limited by the peak radiated power that is determined by hardware configuration references. One of the critical tasks for the future satellite-terrestrial cooperative networks is to maximize the total ergodic  throughput under the power constraints as follows 
\begin{equation} \label{Prob:NMSEk}
	\begin{aligned}
		& \underset{\{ \rho_k \geq 0 \}, \mathcal{Q} }{\mathrm{maximize}}
		&&  \sum\nolimits_{k \in \mathcal{Q}} R_{k} \\
		& \,\,\mathrm{subject \,to}
		&&  \rho_k \leq P_{\max,k}, \forall k,\\
		&&& \mathcal{Q} \subseteq \mathcal{K},
	\end{aligned}
\end{equation}
where $P_{\max,k}$ represents the maximum transmit power that each data symbol can be assigned by user~$k$. We stress that problem~\eqref{Prob:NMSEk} can be applied for the network with arbitrary combining methods once the throughput in \eqref{eq:Rkv1} with the SINR value in \eqref{eq:SINRk} is exploited.  This paper focuses on the sum throughput optimization with the closed-form expression in \eqref{eq:RateMRC} since the optimal power coefficients are a long-term solution that only needs to be updated as the channel statistics vary. Problem~\eqref{Prob:NMSEk} has the continuous objective function together with a compact feasible domain obtained for a fixed active user set $\mathcal{Q}$. Consequently, the globally optimal solution exists and might be obtained if all the possibilities of the active user set $\mathcal{Q}$ are investigated. Combined with the inherent non-convex objective function, the global optimum is, unfortunately, nontrivial to obtain, especially for large-scale systems with many APs and users. To cope with this matter, the throughput in \eqref{eq:RateMRC} should be attained by the signal transmission of an analogous single-input single-output (SISO) system as follows
\begin{equation}
	\tilde{y}_k = \tilde{\rho}_k \left( \|\bar{\mathbf{g}}_k\|^2 +  p \tau_p \mathrm{tr}(\mathbf{R}_k \pmb{\Phi}_k \mathbf{R}_k)  +  \sum\nolimits_{m=1}^M \gamma_{mk} \right) x_k + \tilde{w}_k,
\end{equation}
where $\tilde{\rho}_k = \sqrt{\rho}_k$ and $x_k$ is the transmitted data symbol with $\mathbb{E}\{ x_k^2 \} = 1$. The additive noise $\tilde{w}_k$ is distributed as $\tilde{w}_k \sim \mathcal{N}(0, \delta_k)$ with $\delta_k = \mathsf{CI}_k +  \mathsf{NI}_k + \mathsf{NO}_k$ and $\mathcal{N}(\cdot, \cdot)$ representing a Gaussian distribution. The network utilizes a combining coefficient $v_k \in \mathbb{R}$ to detect the desired signal from user~$k$ as
\begin{multline} \label{eq:DecodedSig}
	\hat{x}_k = v_k \tilde{y}_k = \tilde{\rho}_k v_k  \left( \|\bar{\mathbf{g}}_k\|^2 +  p \tau_p \mathrm{tr}(\mathbf{R}_k \pmb{\Phi}_k \mathbf{R}_k)  +  \sum\nolimits_{m=1}^M \gamma_{mk} \right)\\ \times x_k + v_k \tilde{w}_k.
\end{multline}
By exploiting the decoded signal in \eqref{eq:DecodedSig}, the mean square error (MSE) of our suggested decoding process is formulated as
\begin{equation}
\begin{split}
& e_k = \mathbb{E}\{ (\hat{x}_k  - x_k )^2 \} \stackrel{(a)}{=}   v_k^2  \delta_k + \\
& \left(\tilde{\rho}_k  v_k \left( \|\bar{\mathbf{g}}_k\|^2 +  p \tau_p \mathrm{tr}(\mathbf{R}_k \pmb{\Phi}_k \mathbf{R}_k )  +  \sum\nolimits_{m=1}^M \gamma_{mk}\right) -1 \right)^2 ,
\end{split}
\end{equation}
where $(a)$ is attained by using \eqref{eq:DecodedSig}. After that, problem~\eqref{Prob:NMSEk} is equivalently converted to the sum MSE optimization problem as follows
\begin{equation} \label{Prob:NMSEkE1}
	\begin{aligned}
		& \underset{\{ \alpha_k \geq 0, \tilde{\rho}_k \geq 0, v_k \}, \mathcal{Q} }{\mathrm{minimize}}
		&&  \sum\nolimits_{k \in \mathcal{Q}} \alpha_{k} e_k - \ln(\alpha_k) \\
		& \,\,\mathrm{subject \,to}
		&& \tilde{\rho}_k^2 \leq P_{\max,k}, \forall k, \\
		&&& \mathcal{Q} \subseteq \mathcal{K},
	\end{aligned}
\end{equation}
in the way that they share the same optimal transmit power solution, say $\tilde{\rho}_k^2  = \rho_k, \forall k,$ at the global optimum, with the proof straightforwardly attained by utilizing the similar methodology as in \cite{van2018large}. Compared to the original problem, we have simplified the complexity matter since the sum MSE optimization is element-wise convex in the sense that if only one variable is considered, at the same time, the remaining are fixed, \eqref{Prob:NMSEkE1} becomes a convex problem for a predetermined set of active users. This attractive property should be exploited to attain a stationary solution to problem~\eqref{Prob:NMSEk} iteratively.
  
\subsection{Iterative Algorithm} \label{SubSec:IterAl}
We first tackle the discrete variable in problem~\eqref{Prob:NMSEkE1} by observing that $\mathcal{Q}$ is explicitly defined when the optimal solution to the transmit power coefficients is available. Alternatively, the throughput of each active user remains if we add the impacts of inactive users into its expression due to the zero transmit powers. Subsequently, one can set $\mathcal{Q} = \mathcal{K}$ at the beginning and reformulate \eqref{Prob:NMSEkE1} into an equivalent form as 
\begin{equation} \label{Prob:NMSEkEv2}
	\begin{aligned}
		& \underset{\{ \alpha_k \geq 0, \tilde{\rho}_k \geq 0, u_k \} }{\mathrm{minimize}}
		&&  \sum\nolimits_{k \in \mathcal{K}} \alpha_{k} e_k - \ln(\alpha_k) \\
		& \,\,\mathrm{subject \,to}
		&& \tilde{\rho}_k^2 \leq P_{\max,k}, \forall k.
	\end{aligned}
\end{equation}
The feasible set of problem~\eqref{Prob:NMSEkEv2} is continuous, and the combinatorial issue is completely solved. We can now exploit the element-wise convexity to find a local optimum. For such, the Lagrangian function to problem~\eqref{Prob:NMSEkE1} is first formulated as
\begin{multline}
\label{eq:LagrangianFunc}
\mathcal{L} =  \sum_{k'' \in \mathcal{K}} (\alpha_{k''} e_{k''} - \ln(\alpha_{k''}))  + \sum_{k''\in \mathcal{K} } \mu_{k''}  (\omega_{k''} \delta_{k''} - \tilde{\rho}_{k''}^2 a_{k''}^2) \\ + \sum_{k \in \mathcal{K}} \lambda_{k''} ( \tilde{\rho}_{k''}^2 - P_{\max,k''} ),
\end{multline}
where $\mu_k$ and $\lambda_k$, for all $k$, are the Lagrange multipliers designed for the SINR constraints and the limited power budget constraints, respectively. We now provide an algorithm to attain a stationary solution to problem~\eqref{Prob:NMSEk} by alternately updating the subsets of different optimization variables as the policy presented in Theorem~\ref{Theorem:ClosedFormSol}.
\begin{theorem} \label{Theorem:ClosedFormSol}
From an initial point $\{\tilde{\rho}_k^{(0)} \}$ in the feasible domain, a stationary solution to problem~\eqref{Prob:NMSEkE1} is achieved by iteratively updating $\{ v_k,\alpha_k, \tilde{\rho}_k\}$. At iteration~$n$, those optimization variables are updated in the following order: 
\begin{itemize}
\item The $v_k$ variables, $\forall k$, are updated as in \eqref{eq:vk}. 
\begin{figure*}
\begin{equation} \label{eq:vk}
v_{k,(n)} = \frac{\tilde{\rho}_{k,(n-1) }\left( \|\bar{\mathbf{g}}_k\|^2 +  p \tau_p \mathrm{tr}(\mathbf{R}_k \pmb{\Phi}_k \mathbf{R}_k )  +  \sum\nolimits_{m=1}^M \gamma_{mk}\right) }{ \tilde{\rho}_{k,(n-1)}^2 \left(\|\bar{\mathbf{g}}_k\|^2 +  p \tau_p \mathrm{tr}(\mathbf{R}_k \pmb{\Phi}_k \mathbf{R}_k )  +  \sum\nolimits_{m=1}^M \gamma_{mk} \right)^2 + \delta_{k,(n-1)}}
\end{equation}
%\hrule
\vspace{-0.8cm}
\end{figure*}
where $\delta_{k,(n-1)}$ is computed as in \eqref{eq:Deltak}.
\begin{figure*}
\begin{multline} \label{eq:Deltak}
\delta_{k,(n-1)} =	\sum\nolimits_{k' \in \mathcal{P}_k \setminus \{k\}} \tilde{\rho}_{k',(n-1)}^2 \left| \bar{\mathbf{g}}_{k}^H \bar{\mathbf{g}}_{k'} + p \tau_p   \mathrm{tr}(\mathbf{R}_{k'} \pmb{\Phi}_k \mathbf{R}_k) +   \sum\nolimits_{m=1}^M  \frac{c_{mk'}  }{c_{mk}  }  \gamma_{mk} \right|^2 +  \sum\nolimits_{k' \notin \mathcal{P}_k } \tilde{\rho}_{k',(n-1)}^2 |\bar{\mathbf{g}}_{k}^H  \bar{\mathbf{g}}_{k'} |^2  + \\ \sum\nolimits_{k' \in \mathcal{K}} \tilde{\rho}_{k',(n-1)}^2  p \tau_p \bar{\mathbf{g}}_{k'}^H \mathbf{R}_k \pmb{\Phi}_k \mathbf{R}_k   \bar{\mathbf{g}}_{k'} + \sum\nolimits_{k' \in \mathcal{K}} \tilde{\rho}_{k',(n-1)}^2 \bar{\mathbf{g}}_{k}^H \mathbf{R}_{k'} \bar{\mathbf{g}}_{k}    + p \tau_p \sum\nolimits_{k' \in \mathcal{K} } \tilde{\rho}_{k',(n-1)}^2   \mathrm{tr}( \mathbf{R}_{k'}\mathbf{R}_k \pmb{\Phi}_k \mathbf{R}_k ) 
 +  \\ \sum\nolimits_{k' \in \mathcal{K}} \sum\nolimits_{m=1}^M \tilde{\rho}_{k',(n-1)}^2  \gamma_{mk} \beta_{mk'}+ \sigma_s^2 \|\bar{\mathbf{g}}_k\|^2 +   p \tau_p \sigma_s^2 \mathrm{tr}(\mathbf{R}_k \pmb{\Phi}_k \mathbf{R}_k) 
 + \sigma_a^2 \sum\nolimits_{m=1}^M \gamma_{mk}
\end{multline}
\hrule
\vspace{-0.5cm}
\end{figure*}
\item The $\alpha_k$ variables, $\forall k$, are updated as
\begin{equation} \label{eq:alphak}
\alpha_{k,(n)} = 1/e_{k,(n)},
\end{equation}
where $e_{k,(n)}$ is computed based on $\{ \tilde{\rho}_{k,(n-1)}\}$ and $\{v_{k,(n)}\}$ as
\begin{multline} \label{eq:ekn}
e_{k,(n)} = \Big(\tilde{\rho}_{k,(n-1)} v_{k,(n)} \big( \|\bar{\mathbf{g}}_k\|^2 +  p \tau_p \mathrm{tr}(\mathbf{R}_k \pmb{\Phi}_k \mathbf{R}_k )  + \\  \sum\nolimits_{m=1}^M \gamma_{mk} \big)-1 \Big)^2 + v_{k,(n)}^2  \delta_{k,(n-1)}.
\end{multline}
\item The $\tilde{\rho}_k$ variables, $\forall k$, are updated as
\begin{equation} \label{eq:rhoktilde}
	\tilde{\rho}_k = \min( \bar{\rho}_{k,(n)}, \sqrt{P_{\max,k}} ),
\end{equation}
where $\bar{\rho}_{k,(n)}$ is computed  as
\begin{multline} \label{eq:barrhok}
	\bar{\rho}_{k,(n)} =   \left( \|\bar{\mathbf{g}}_k\|^2 +  p \tau_p \mathrm{tr}(\mathbf{R}_k \pmb{\Phi}_k \mathbf{R}_k )  +  \sum\nolimits_{m=1}^M \gamma_{mk} \right) \times \\
 \alpha_{k,(n)}  v_{k,(n)} / t_{k,(n)},
\end{multline}
with $t_{k,(n)}$  defined as 
\begin{multline} \label{eq:tkn}
t_{k,(n)} =   \alpha_{k,(n)} \Big( \|\bar{\mathbf{g}}_k\|^2 +  p \tau_p \mathrm{tr}(\mathbf{R}_k \pmb{\Phi}_k \mathbf{R}_k )  +  \sum\nolimits_{m=1}^M \gamma_{mk} \Big)^2 \\
 \times v_{k,(n)}^2  +  \sum\nolimits_{k'' \in \mathcal{K}}  \alpha_{k'',(n)} v_{k'',(n)}^2 \times \\
	 \Big(  p \tau_p \bar{\mathbf{g}}_{k}^H \mathbf{R}_{k''} \pmb{\Phi}_{k''} \mathbf{R}_{k''}  \bar{\mathbf{g}}_{k} +  \bar{\mathbf{g}}_{k''}^H \mathbf{R}_{k} \bar{\mathbf{g}}_{k''}  +  p \tau_p   \mathrm{tr}( \mathbf{R}_{k} \mathbf{R}_{k''} \pmb{\Phi}_{k''} \mathbf{R}_{k''}  ) \\ +  \sum\nolimits_{m=1}^M  \gamma_{mk''} \beta_{mk} \Big) +  \sum\nolimits_{k'' \notin \mathcal{P}_k }  \alpha_{k'', (n)}  v_{k'',(n)}^2  |\bar{\mathbf{g}}_{k''}^H \bar{\mathbf{g}}_k|^2  \\
  +   \sum\nolimits_{k'' \in \mathcal{P}_{k} \setminus \{k\} } \alpha_{k'',(n)} v_{k'',(n)}^2  \Big| \bar{\mathbf{g}}_{k''}^H \bar{\mathbf{g}}_{k} + p \tau_p   \mathrm{tr}(\mathbf{R}_{k} \pmb{\Phi}_{k''} \mathbf{R}_{k''}) \\ +   \sum\nolimits_{m=1}^M  \frac{c_{mk}  }{c_{mk''}  }  \gamma_{mk''} \Big|^2.
\end{multline}
\end{itemize}
If we denote the fixed point solution attained by the above iterative algorithm as $\{ v_k^{\ast}, \alpha_k^{\ast}, \tilde{\rho}_k^{\ast} \}$, then $\{ \rho_k^\ast \}$ is a stationary solution to problem~\eqref{Prob:NMSEk}.
\end{theorem}
\begin{proof}
See Appendix~\ref{Appendix:ClosedFormSol} for the detailed proof.
\end{proof}
 The proposed iterative approach to obtain a stationary solution to the joint power allocation and active user set is summarized in Algorithm~\ref{Algorithm1}. For a given power coefficients, say $ 0 \leq \rho_{k, (0)} \leq P_{\max, k}, \forall k$,  in the feasible domain, we compute the related optimization  variables $\tilde{\rho}_{k,(0)} = \tilde{\rho}_{k,(0)}, \forall k$. In iteration~$n$, the beamforming variables $v_{k,(n)}, \forall k,$ are updated by exploiting the closed-form expression in \eqref{eq:vk} with the square root of the  coefficients from the previous iteration and $\delta_{k,(n-1)}$ computed as in \eqref{eq:Deltak}. After that, the weighted variables $\alpha_{k,(n)}, \forall k,$ are updated by exploiting the closed-form expression in \eqref{eq:alphak} with $e_{k,(n)}$ computed by as in \eqref{eq:ekn}. Algorithm~\ref{Algorithm1} then updates the optimization variables $\tilde{\rho}_{k,(n)}, \forall k,$ by utilizing \eqref{eq:rhoktilde} with $\bar{\rho}_{k,(n)}$ given in \eqref{eq:barrhok} and $t_{k,(n)}$ given in \eqref{eq:tkn}. The CPU can terminate Algorithm~\ref{Algorithm1} when the total throughput has a small variation between the two consecutive iterations as follows 
\begin{equation} \label{eq:StoppingCriterion}
\left| \sum\nolimits_{k \in \mathcal{K}} R_{k,(n)}  - \sum\nolimits_{k \in \mathcal{K}} R_{k,(n-1)} \right| \leq \epsilon.
\end{equation}
We stress that Theorem~\ref{Theorem:ClosedFormSol} gives twofold: First, from an initial point of the power domain, the proposed iterative algorithm will converge to a stationary solution of problem~\eqref{Prob:NMSEkEv2} since each  optimization variable is computed in closed form based on the first-order derivative of the Lagrangian function while the others are fixed. Second, it is a low computational complexity design where all the optimization variables are computed in closed form. Consequently, the proposed iterative algorithm enables joint transmit power control and scheduled user optimization for integrated satellite-terrestrial cell-free massive MIMO systems with many APs and users.
\begin{corollary} 
From the stationary solution $\{ \tilde{\rho}_k^\ast \}$, we set $\rho_k^\ast = (\tilde{\rho}_k^\ast)^2$, for all $k$, and the optimized scheduled user set $\mathcal{Q}^{\ast}$ is explicitly defined as
\begin{equation}
	\mathcal{Q}^\ast = \{  k | \rho_k^\ast > 0, k \in \mathcal{K} \},
\end{equation}
and the unscheduled user set is formulated as $	\bar{\mathcal{Q}}^\ast = \mathcal{K} \setminus \mathcal{Q}^\ast.$
\end{corollary}
Alternatively, we have demonstrated an effective way to obtain both the power allocation to all the users and define the scheduled and unscheduled subsets. 
\begin{algorithm}[t]
	\caption{Alternating optimization approach for \eqref{Prob:NMSEkEv2}} \label{Algorithm1}
	\textbf{Input}:  Channel statistics $\{ \bar{\mathbf{g}}_k, \mathbf{R}_k, \pmb{\Phi}_k, \gamma_{mk}, \beta_{mk} \}$ ; Maximum power levels  $P_{\max,k}, \forall k$; Select initial values $\tilde{\rho}_{k,(0)}  \forall k$; Set up $n=0$ and tolerance $\epsilon$.
	\\
	\textbf{While} \textit{Stopping criterion \eqref{eq:StoppingCriterion} is not satisfied} \textbf{do}
	\begin{itemize}
		\item[1.] Set $n=n+1$.
		\item[2.] Update $v_{k,(n)}$ for all $k$ by \eqref{eq:vk} where each $\delta_{k,(n-1)}$ is computed as in \eqref{eq:Deltak}.
		\item[3.] Update $\alpha_{k,(n)}$ for all $k$ by \eqref{eq:alphak} where each $e_{k,(n)}$ is computed as in \eqref{eq:ekn}.
		\item[4.] Update $\tilde{\rho}_{k,(n)}$ for all $k$ by  \eqref{eq:rhoktilde} where each $\bar{\rho}_{k,(n)}$ is computed as in \eqref{eq:barrhok} with $t_{k,(n)}$ given in \eqref{eq:tkn}.
		\item[5.] Store the current solution $\tilde{\rho}_{k,(n)}$.
	\end{itemize}
   \textbf{End while}\\
	\textbf{Output}: The stationary solution $\tilde{\rho}_{k}^{\ast} = \tilde{\rho}_{k,(n)}$, $\forall k$. %\vspace*{-0.2cm}
\end{algorithm}
\section{Deep Learning Framework} \label{Sec:DL}
This section describes a $K$-user interference channel in satellite-terrestrial networks as a heterogeneous graph. Then, using the framework of GNN, we propose Satellite Heterogeneous Graph Neural Network (SHGNN) that learns to predict the policy producing optimal power allocation policy deploying channel statistics in an unsupervised fashion.
\subsection{Graphical Representation}
A graph can be formulated as a tuple $\mathcal{G} = \{\mathcal{V},\mathcal{E}\}$, where $\mathcal{V}$ stands for the set of vertices and $\mathcal{E}$ denotes the set of edges. Generally, vertices and edges can belong to different types. We denote $\mathcal{A}$ as the set of vertex types, where $\mathcal{R}$ is the set of edge types. Graph $\mathcal{G}$ is a homogeneous graph (HomoGraph) if $|\mathcal{A}|=|\mathcal{R}| = 1$, otherwise it is a heterogeneous graph (HetGraph).
In this paper, we formulate the sum throughput optimization problem \eqref{Prob:NMSEk} with the transmit powers and active user set as a learning problem over the following HetGraph with three types of vertices, i.e., $|\mathcal{A}| = 3$, as described as follows:

\textbf{Vertexes and Edges:}
\begin{itemize}
	\item Each AP or each user or the satellite is a vertex.
	\item Each channel link between APs and users or  between the satellite and users is an edge.
\end{itemize} 

\textbf{Features:}
\begin{itemize}
	\item The vertex feature of each user is the available transmit power, i.e., $P_{\max,k}$. APs and the satellite have no feature.
	\item The feature of the edge between AP $m$ and divice $k$ is the large-scale fading coefficient $\beta_{mk}$. The feature of the edge between the satellite and user~$k$ is the LoS channel $\bar{\mathbf{g}}_k$ and the correlation matrix $\mathbf{R}_k$.
\end{itemize}
\begin{figure}[t]
	\centering
	\includegraphics[trim=9.5cm 5.5cm 4.8cm 5.5cm, clip=true, width=3in]{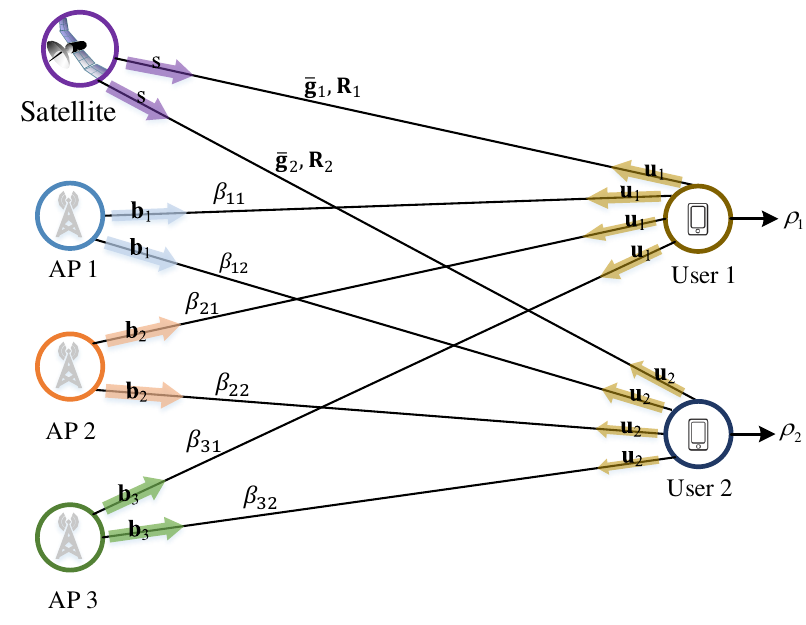} \vspace*{0cm}
	\caption{Heterogeneous graph representation of the considered satellite-terrestrial cooperative network with $M = 3$ and $K = 2$.}
	\label{FigSysModel}
	\vspace*{-0.5cm}
\end{figure}
We entitle this HetGraph as a heterogeneous wireless interference graph (HWIG), where each AP and each user are linked with each other, and each user is also linked to the satellite. We illustrate the graph representation of the considered system in Fig. The designed HWIG can capture the permutation equivariance property of the power allocation optimization problem \eqref{Prob:NMSEk} \cite{Shen2021GNND2D}. Specifically, if the order of users is permuted in the problem, the optimal allocated power should be permuted correspondingly. Generally, simple neural networks such as DNN or CNN can not exploit permutation equivariance property since they can not capture the interaction between each entity in the network. Therefore, in this paper, we apply a HetGNN to learn over the designed HWIG that can exploit the permutation equivariance property \cite{Zhang2019HetGNN}.
\subsection{Heterogeneous Graph Neural Networks (HetGNN) for Integrated Satellite-Terrestrial Cell-Free Massive MIMO}
The HetGNN has been proposed to learn over HetGraph data. For a conventional HetGNN with multiple layers, each vertex updates its hidden state relying on the cascaded information for its neighbor vertices and edges connected to that vertex. Specifically, denote $\mathbf{d}_{m}^{(l)}$ as the output of the $m$-th vertex at the $l$-th layer, each HetGNN layer updates its state via the two steps \cite{Zhang2019HetGNN}:
\begin{enumerate}
	\item \textbf{Aggregation:} The $m$-th vertex employs a neural network to collect its neighbor vertices' output from the last layer and the feature of edges connecting them as
	\begin{equation} \label{eq:Aggre}
		\mathbf{a}_{m,t}^{(l)} = \mathrm{PL}_{n\in \mathcal{N}_t{(m)}}\left(q_1(\mathbf{d}_n^{(l-1)},\mathbf{e}_{mn},\mathbf{W}_{1,t}^{(l)})\right), t \in \mathcal{A},
	\end{equation}
	where $\mathrm{PL}_{n\in \mathcal{N}_t{(m)}}$ is a pooling function, $\mathcal{N}_t{(m)}$ is the set of neighbors of the $m$-th vertex with type $t$, $q_1(.)$ is a neural network used for the aggregation, $\mathbf{W}_{1,t}^{(l)}$ denotes the weights of $q_1(.)$, and $\mathbf{e}_{mn}$ is the feature of the edge $(m,n)$.
	\item \textbf{Combination:} After the aggregated information is obtained at each vertex, a neural network should be applied to manipulate  information and produce the output at each vertex. More specifically, the output is obtained at each vertex as
	\begin{equation} \label{eq:Combine}
		\mathbf{d}_m^{(l)} = q_2\left( \mathbf{d}_m^{(l-1)},\{\mathbf{a}_{m,t}^{(l)},t\in \mathcal{A}\},\mathbf{W}_{2,t}^{(l)}\right),
	\end{equation}
where $q_2(\cdot)$ denotes the combination neural network including learnable parameter $\mathbf{W}_{2,t}^{(l)}$.
\end{enumerate}

We can see from \eqref{eq:Aggre} and \eqref{eq:Combine} that the order of vertices does not affect the output at each vertex since the information is processed independently at each vertex. Furthermore, because each vertex with the same type is processed by the same neural networks, i.e., $q_1(.,.,\mathbf{W}_{1,t}^{(l))})$ and $q_2(.,.,\mathbf{W}_{2,t}^{(l))})$, the output dimension is invariant with the number of APs/users. Therefore, unlike MLPs or CNNs where data dimension must be the same during the training and inference phase, the proposed HetGNN can easily be generalized to different problem sizes.
\subsection{Implementation of SHGNN}
In this subsection, we present the detailed design of SHGNN to learn about the HWIG. To distinguish the outputs of three types of vertices in HWIG, we denote $\mathbf{b}_{m}^{l}$, $\mathbf{u}_{k}^{l}$, and $\mathbf{s}^{l}$ as the output of AP $m$, user~$k$ and the satellite in the $l$-th layer. The computational procedure of the proposed HetGNN comprises three phases as follows.
\subsubsection{Feature Initialization} In the initialization phase, we design the feature for each type of vertex and edge for SHGNN. Firstly, the input feature from the user vertices is the transmit power at each user, i.e., $\mathbf{u}_{k}^{0} = P_{\max,k}$. Since vertices of the APs and the satellite do not have any feature, we set 1 as the input of them, i.e., $\mathbf{b}_{m}^{0}=\mathbf{s}^{0}=1, \forall m$. For the edges between the APs and the users, the large-scale fading coefficients are used as features as $\mathbf{e}_{\mathrm{AP}_m-\mathrm{user}_k} = \beta_{m,k}$. Finally, the LoS channels and the correlation matrices of the space links are used as the feature of the corresponding edges as
\begin{multline}
\mathbf{e}_{\mathrm{satellite}-\mathrm{user}_k}= [\mathrm{vec}(\mathrm{Re}\{\bar{\mathbf{g}}_k\})^T,\mathrm{vec}(\mathrm{Im}\{\bar{\mathbf{g}}_k\})^T, \\ \mathrm{vec}(\mathrm{Re}\{{\mathbf{R}}_k\})^T,\mathrm{vec}(\mathrm{Im}\{{\mathbf{R}}_k\})^T]^T,
\end{multline}
where $\mathrm{vec}(.)$ denotes the vectorization operator. Besides,  $\mathrm{Re}(\mathbf{A})$ and $\mathrm{Im}(\mathbf{A})$ denotes the real and imaginary parts of matrix $\mathbf{A}$.
\subsubsection{Data Processing} SHGNN updates its information with aggregation and combination steps as shown in \eqref{eq:Aggre} and \eqref{eq:Combine}. Specifically, the update of SHGNN consists of three parts as

\textbf{APs aggregating information from users}
\begin{equation}
	\begin{split}
		\mathrm{Aggregate:} \quad &\mathbf{a}_{m,\mathrm{AP}}^{(l)} = \mathrm{MEAN}_{k\in\mathcal{K}}\left\{\mathrm{MLP1}(\mathbf{e}_{\mathrm{AP}_m-\mathrm{user}_k},\mathbf{b}_m^{(l-1)})\right\}, \\
		\mathrm{Combine:} \quad  & \mathbf{b}_m^{(l)} = \mathrm{ReLU} \left(\mathrm{MLP2}(\mathbf{b}_m^{(l-1)},\mathbf{a}_{m,\mathrm{AP}}^{(l)})\right),
	\end{split}	
\end{equation}

\textbf{Satellite aggregating information from users}
\begin{equation}
	\begin{split}
		\mathrm{Aggregate:} \quad &\mathbf{a}_{\mathrm{satellite}}^{(l)} = \\
  &\mathrm{MEAN}_{k\in\mathcal{K}}\left\{\mathrm{MLP3}(\mathbf{e}_{\mathrm{satellite}-\mathrm{user}_k},\mathbf{s}^{(l-1)})\right\}, \\
		\mathrm{Combine:} \quad  & \mathbf{s}^{(l)} = \mathrm{ReLU} \left(\mathrm{MLP4}(\mathbf{s}^{(l-1)},\mathbf{a}_{\mathrm{satellite}}^{(l)})\right),
	\end{split}	
\end{equation}

\textbf{Users aggregating information from APs and satellite}
\begin{equation}
	\begin{split}
		\mathrm{Aggregate:} \quad &\mathbf{a}_{k,\mathrm{user-AP}}^{(l)} = \\
     &\mathrm{MEAN}_{m\in\mathcal{M}}\left\{\mathrm{MLP5}(\mathbf{e}_{\mathrm{AP}_m-\mathrm{user}_k},\mathbf{u}_k^{(l-1)})\right\}, \\
		&\mathbf{a}_{k,\mathrm{user-satellite}}^{(l)} = \mathrm{MLP6}(\mathbf{e}_{\mathrm{satellite}-\mathrm{user}_k},\mathbf{u}_k^{(l-1)}), \\
		\mathrm{Combine:} \quad  & \mathbf{u}_k^{(l)} = \\
  &\mathrm{ReLU} \left(\mathrm{MLP7}(\mathbf{u}_k^{(l-1)},\mathbf{a}_{k,\mathrm{user-AP}}^{(l)},\mathbf{a}_{k,\mathrm{user-satellite}}^{(l)})\right),
	\end{split}	
\end{equation}
where $\mathrm{ReLU}(.)$ is the ReLu activation function, $\mathrm{MEAN}(\{.\})$ is the pooling function that calculates the mean value of a set. At the last HetGNN layer, the output of user vertices will be processed by an MLP to obtain the optimal power allocation vector. A Sigmoid activation layer ensures the predicted power vector satisfies the transmit power constraint. Specifically, the optimal allocated power is obtained as $\mathbf{p}^* = \mathbf{P}_{\max} \odot \sigma\left(\mathrm{MLP8}(\mathbf{u}^{(D)})\right)$, where $\sigma(.)$ denote Sigmoid activation function, $\mathbf{P}_{\max} = [P_{\max,1},\cdots,P_{\max,K}]$, $\mathbf{p}^* = [p_1^*,\cdots,p_K^*]$ is the optimal allocated power vector, $\mathbf{u}^{(D)} = [\mathbf{u}_1^{(D)},\cdots,\mathbf{u}_K^{(D)}]$, and $D$ is the number of HetGNN layers. Finally, the loss function adopted to train the neural network is the negative sum rate defined as
\begin{equation}
	\mathcal{L} = -\mathbb{E}\left\{B\left( 1 - \frac{\tau_p}{\tau_c} \right) \sum_{k=1}^{K} \log_2 ( 1 + \mathrm{SINR}_k(\mathbf{p}^*(\theta)) )\right\},
\end{equation}
where $\theta$ denotes parameters of SHGNN, $\mathrm{SINR}_k$ is the effective SINR of user $k$ calculated as in \eqref{eq:ClosedSINR}. It is noteworthy that SHGNN is trained unsupervised without requiring any labels. We stress that the proposed SHGNN can be applied for other objective functions in the considered system by simply replacing the loss function with the one that needs to be optimized.
\section{Numerical Results} \label{Sec:NumericalResults}
We now evaluate the analysis and optimal power allocation discussed in Section~\ref{Sec:UplinkPerfAna}--\ref{Sec:DL} for a system with $40$ APs spreading in a square area of $16$~[km$^2$] that serves a various number of users. All the devices are mapped into a 3D Cartesian coordinate system. An NGSO satellite is  at position $(300, 300, 400)$~[km]. The satellite antenna is fabricated by a rectangular array with $N_V = N_H  = 5$. The operating  bandwidth is set to $B= 20$~[MHz], and the frequency of the carrier wave is $3$~[GHz]. Meanwhile, the noise figures at the APs and the satellite are $1.2$~[dB] and $4$~[dB], respectively. Moreover, the large-scale fading  between AP~$m$ and user~$k$, measured in [dB], is defined as 
%\begin{equation} \label{eq:betamk}
$\beta_{mk} = G_m  + G_k - 8.5 - 38.63 \log_{10} (d_{mk}/d_0) - 20 \log_{10} (f_c) + z_{mk}$,
%\end{equation}
where $z_{mk} \sim \mathcal{CN}(0, \delta_m^2)$ represents the influence of shadow fading having zero mean and $\delta_m$ [dB] denoting its standard derivation. The slope distance thresholds are $d_0$~[m] and $d_1$~[m], respectively. In addition, $G_m$~[dBi] and $G_k$~[dBi] represent the antenna gain at AP~$m$ and user~$k$, respectively. Here, we select $G_k = G_m = 10$~[dBi]. Besides, the large-scale fading, measured in [dB], between the satellite and user ~$k$ is given as $\beta_k = G_k + G - 32.45 - 20 \log_{10} (d_k /d_0) - 20 \log_{10} f_c + z_k$, where $G$ is the antenna gain at the satellite \cite{van2022space} and $z_k$ stands for the shadow fading defined by a log-normal distribution with its variance 
being a function of the elevation angle, the channel condition, and the carrier frequency \cite{lin20215g}. 
\begin{figure*}[t]
	\begin{minipage}{0.48\textwidth}
		\centering
		\includegraphics[trim=0.6cm 0.0cm 1.1cm 0.6cm, clip=true, width=3.2in]{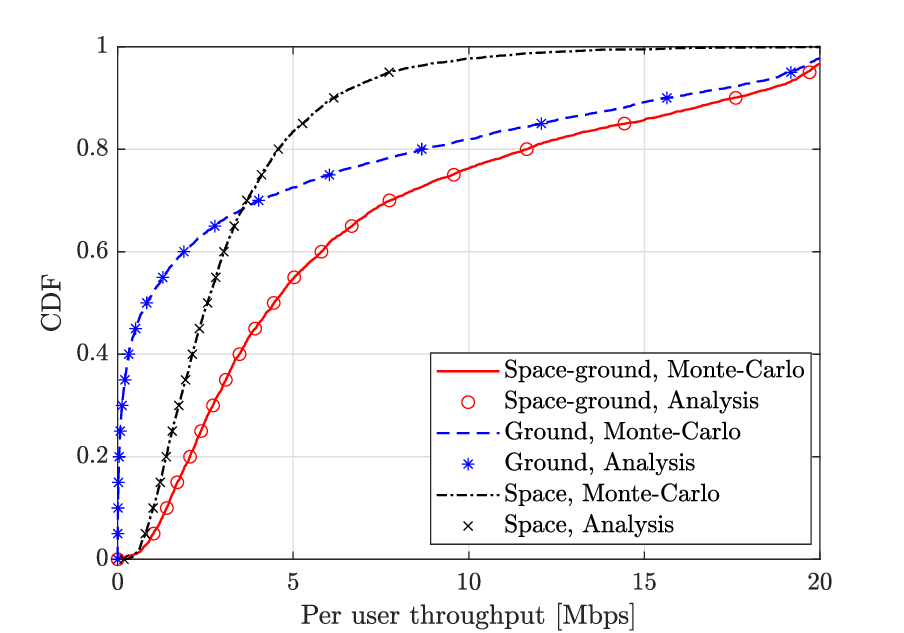} \vspace*{-0.2cm}
		\caption{CDF of the per user throughput [Mbps] utilizing Monte Carlo simulations versus the analyses with $K=20, \tau_c = 10000,$ and $\tau_p = K/2$. } \label{Fig:MCCFPerUser}
		\vspace*{-0.2cm}
	\end{minipage}
	\hfil
	\begin{minipage}{0.48\textwidth}
		\centering
		\includegraphics[trim=0.6cm 0.0cm 1.4cm 0.6cm, clip=true, width=3.2in]{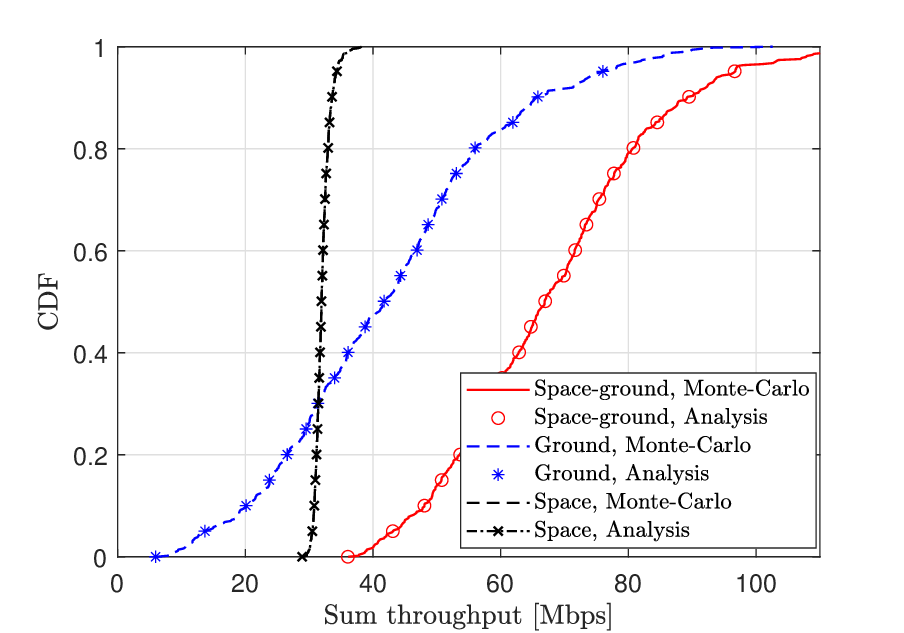} \vspace*{-0.2cm}
		\caption{CDF of the sum throughput [Mbps] utilizing Monte Carlo simulations versus the analyses with $K=20, \tau_c = 10000,$ and $\tau_p = K/2$.} \label{Fig:MCCFSum}
		\vspace*{-0.2cm}
	\end{minipage}
\end{figure*}
\begin{figure*}[t]
	\begin{minipage}{0.48\textwidth}
		\centering
		\includegraphics[trim=0.6cm 0.0cm 1.4cm 0.6cm, clip=true, width=3.2in]{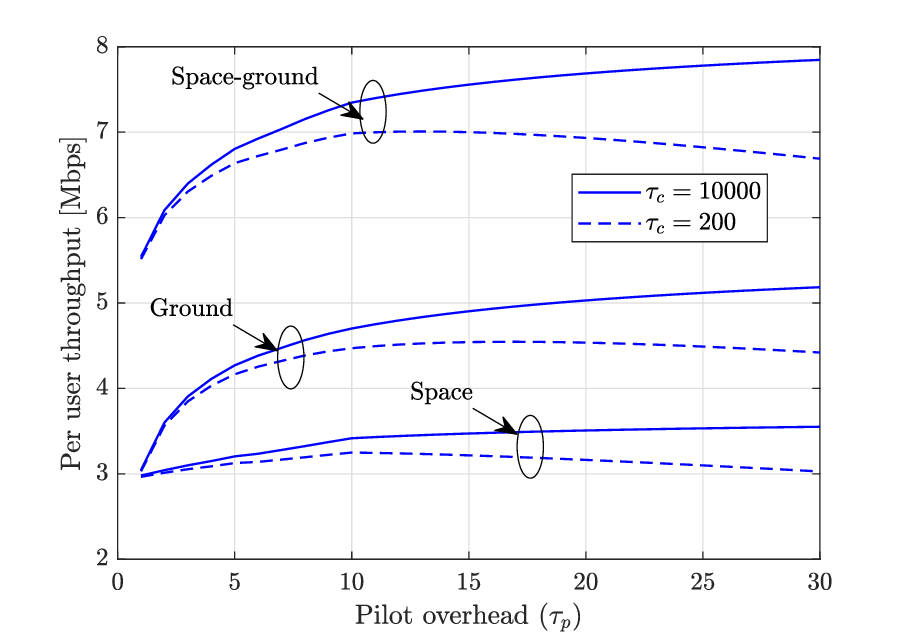} \vspace*{-0.2cm}
		\caption{Per user throughput versus the different number of orthogonal pilot signals with $K = 10$.} \label{Fig:PerUserRatetaup}
		\vspace*{-0.2cm}
	\end{minipage}
	\hfil
	\begin{minipage}{0.48\textwidth}
		\centering
		\includegraphics[trim=0.6cm 0.0cm 1.1cm 0.6cm, clip=true, width=3.2in]{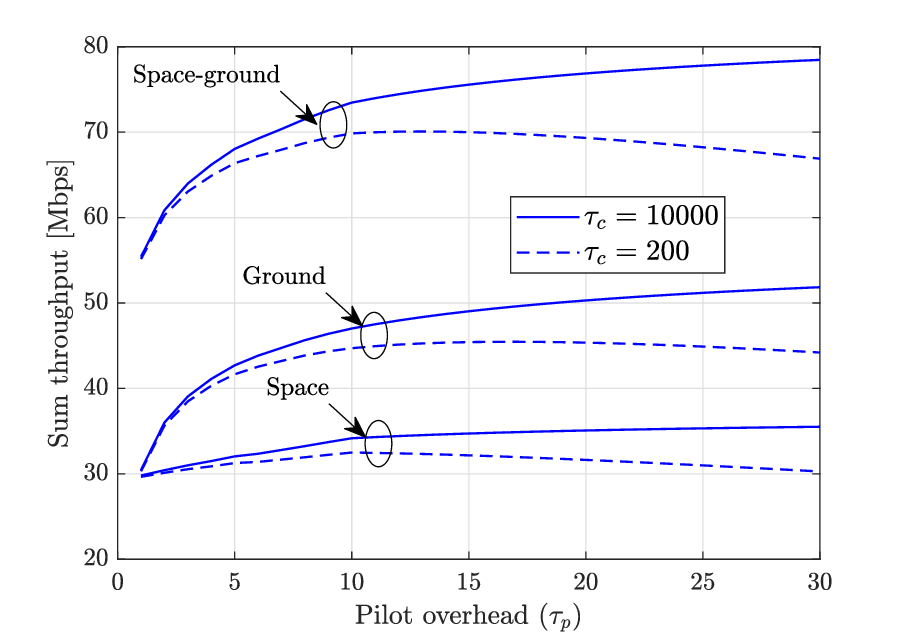} \vspace*{-0.2cm}
		\caption{Sum throughput versus the different number of orthogonal pilot signals with $K = 10$.} \label{Fig:SumRatetaup}
		\vspace*{-0.2cm}
	\end{minipage}
\end{figure*}
\subsection{All Users are Scheduled}
In Fig.~\ref{Fig:MCCFPerUser}, we display the cumulative distribution function (CDF) of the per-user throughput by the three different benchmarks comprising the space-ground cooperative network that comprises communication from both satellite and APs (denoted as Space-ground); the space network that comprises communication from the satellite only (denoted as Space); and the terrestrial system that comprises communication from the APs only (denoted as Ground). Monte-Carlo simulations  entirely overlap with the analytical results, confirming the correctness of our closed-form expression of the uplink ergodic throughput obtained in Theorem~\ref{Theorem:ClosedForm} for space-ground cooperative networks and its typical scenarios in which either only the satellite or APs are exploited. Even though the APs can offer better per-user throughput than the satellite if the LoS channels from the space links are weak, the latter outperforms the former $3.3 \times$ on the median. By coherently processing the received signals, the satellite-terrestrial cooperative system inherits the benefits of multiple observations for a constructive combination at the CPU. Specifically, the cooperative network provides $5.9 \times$ better median per user throughput than the baseline. 

In Fig.~\ref{Fig:MCCFSum}, we show the CDF of sum throughput by the considered benchmarks. Monte-Carlo simulations match well with the analytical results for all the considered various user locations and shadow fading realizations. Besides, the space network with a single satellite yields an average sum throughput of about $32.1$~[Mbps]. More fluctuations, the ground network only offers $42.5$~[Mbps], which is  $1.3\times$ better than only utilizing a satellite to serve many users. The space-ground cooperative network offers a sum throughput of about $67.9$~[Mbps] that surpasses the space network $2.1\times$. At the $95\%$-likely, the APs only can provide $14.5$~[Mbps] to all the users in total, while that of the space network is $30.5$~[Mbps] that implies a superiority of $2.1\times$. Meanwhile, the considered integragated cooperative network can improve the sum throughput of $3.0 \times$. Analytical and numerical results confirm the advantages of the space-ground cooperative systems in upgrading user throughput in the coverage area.

In Figs.~\ref{Fig:PerUserRatetaup} and \ref{Fig:SumRatetaup}, we exploit the effects of channel estimation overhead in terms of  the per user and sum throughput, respectively. The space-ground cooperative network outperforms the remaining benchmarks significantly over all the considered length of pilot signals. As $\tau_c = 200$ and $\tau_p = 10$, the per  throughput offered by the space-ground, ground, and space network is $7.0$~[Mbps], $4.5$~[Mbps], and $3.2$~[Mbps], respectively. It implies the improvement of our considered cooperative network of about $1.6 \times$ and $2.2\times$ compared to the system with either the APs or satellite only. Increasing the number of symbols in each coherence interval for the pilot training phase improves both the sum and per user throughput as $\tau_p$ is still small since the network can exploit more orthogonal pilot signals and therefore the channel estimation quality is improved drastically. Nonetheless, for the fast fading channels with short coherence time, e.g., $\tau_c = 200$, as the pilot overhead exceeds the number of users, the throughput decreases due to the short coherence time allocated to the data transmission. In contrast, if the propagation channels are more stable, e.g., $\tau_c = 100000$, the  throughput consistently increases under the parameter settings.

\begin{figure*}[t]
	\begin{minipage}{0.48\textwidth}
		\centering
		\includegraphics[trim=0.6cm 0.0cm 1.1cm 0.6cm, clip=true, width=3.2in]{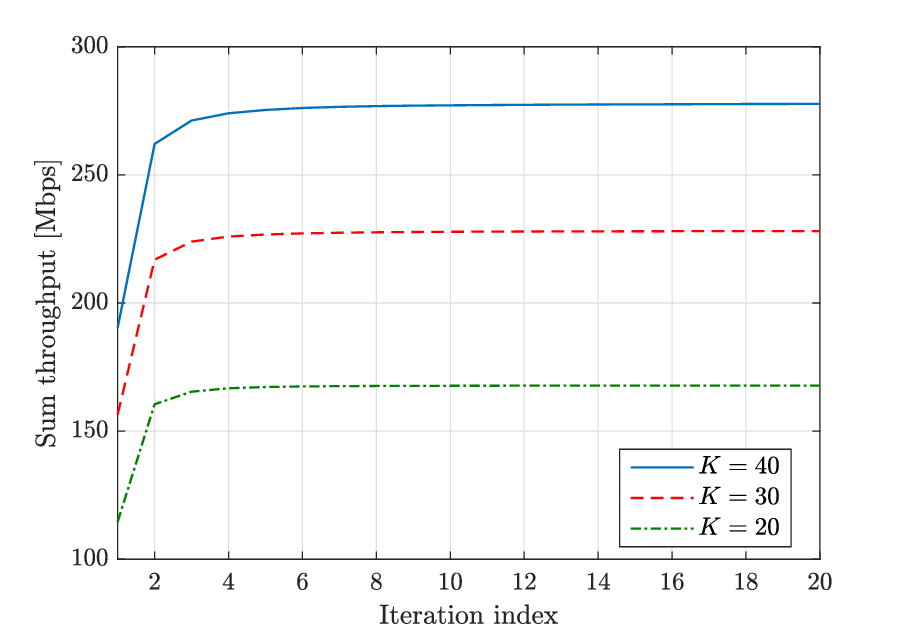} \vspace*{-0.2cm}
		\caption{Convergence of Algorithm~\ref{Algorithm1} along iterations with  a different number of users, $\tau_c = 10000$, and $\tau_p = K/2$.} \label{Fig:Convergence}
		\vspace*{-0.2cm}
	\end{minipage}
	\hfil
	\begin{minipage}{0.48\textwidth}
		\centering
		\includegraphics[trim=0.6cm 0.0cm 1.4cm 0.6cm, clip=true, width=3.2in]{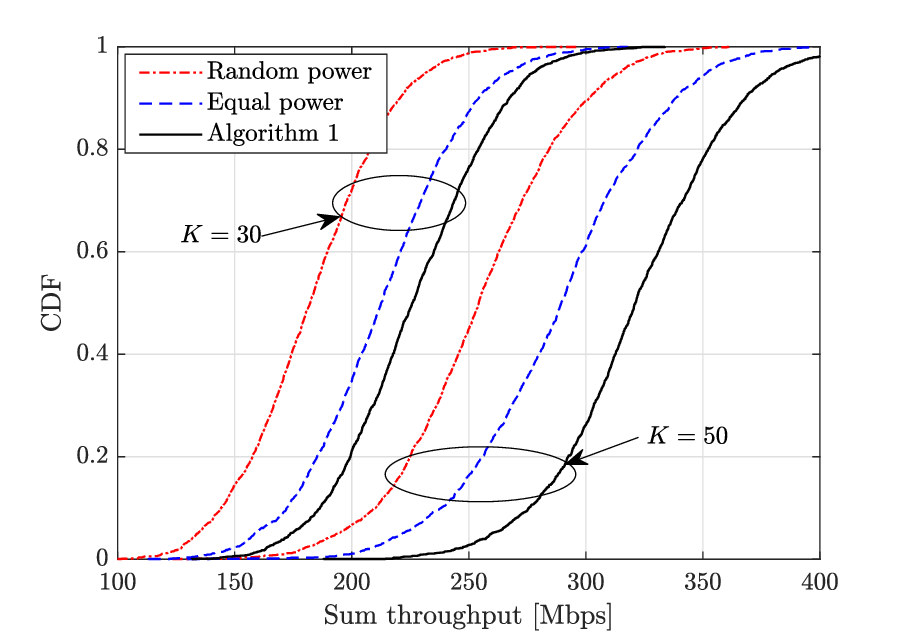} \vspace*{-0.2cm}
		\caption{CDF of the sum throughput [Mbps]  with different power allocation strategies, $\tau_c = 10000$, and $\tau_p = K/2$.} \label{Fig:SumRate}
		\vspace*{-0.2cm}
	\end{minipage}
\end{figure*}
\begin{figure*}[t]
	\begin{minipage}{0.48\textwidth}
		\centering
		\includegraphics[trim=0.6cm 0.0cm 1.1cm 0.6cm, clip=true, width=3.2in]{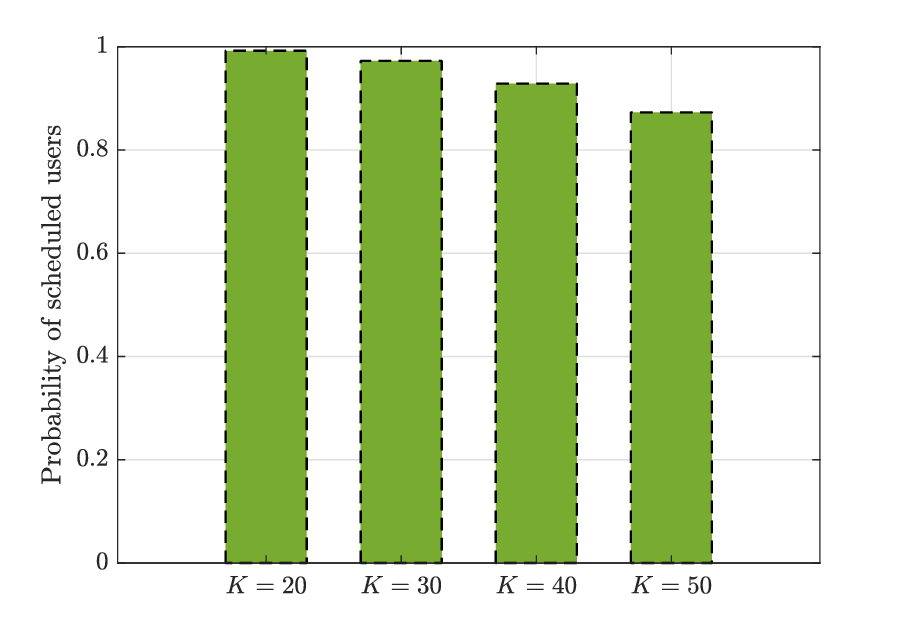} \vspace*{-0.2cm}
		\caption{Probability of scheduled users versus the different number of available users with $\tau_c = 10000$ and $\tau_p = K/2$.} \label{Fig:Scheduled}
		\vspace*{-0.2cm}
	\end{minipage}
	\hfil
	\begin{minipage}{0.48\textwidth}
		\centering
		\includegraphics[trim=0.4cm 0.0cm 1.4cm 0.6cm, clip=true, width=3.2in]{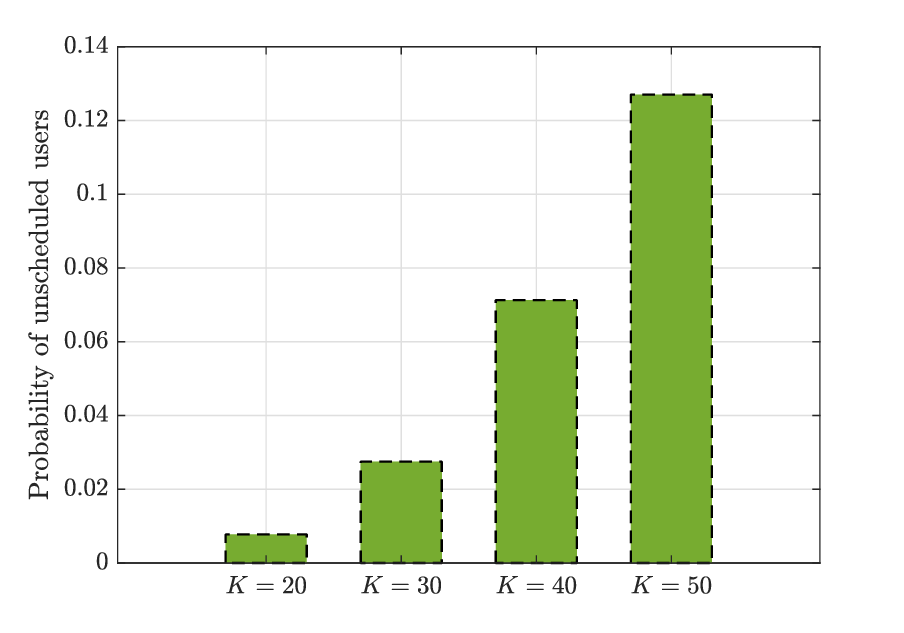} \vspace*{-0.2cm}
		\caption{Probability of unscheduled users versus the different number of available users with $\tau_c = 10000$ and $\tau_p = K/2$.} \label{Fig:Unscheduled}
		\vspace*{-0.2cm}
	\end{minipage}
\end{figure*}
\subsection{Scheduled Users versus Unscheduled Users}
In Fig.~\ref{Fig:Convergence}, the convergence property of Algorithm~\ref{Algorithm1}  is plotted with the different number of available users. The numerical results unveil that the proposed power allocation and user scheduling algorithm converges very fast after less than $10$ iterations. We further observe superior improvements in the solution along the iteration index.  For example, the initial power allocation only provides the sum throughput of $114.6$~[Mbps] to $20$ available users. Nonetheless, the stationary solution provides the sum throughput up to $167.8$~[Mbps], related to an improvement of $1.5\times$ compared to the initialization. This is because a number of users have been unscheduled to avoid dramatic power consumption under unfavorable positions with weak channel gains. 

In Fig.~\ref{Fig:SumRate}, we plot the CDF of the sum throughput by deploying three power allocation strategies consisting of $i)$ Algorithm~\ref{Algorithm1}; $ii)$ the random power allocation where the transmit power coefficients are uniformly distributed as $0\leq \rho_k \leq P_{\max,k}, \forall k$ \cite{sun2018learning}; and $iii)$ the equal power allocation where all the power coefficients are allocated with the maximum level, i.e., $\rho_k = P_{\max,k}, \forall k$ \cite{van2020power}.  While Algorithm~\ref{Algorithm1} deactivates several users with weak channel gains  to reduce mutual interference, the remaining benchmarks admit all $K$ users to the network. For a system with $30$ users, the sum throughput is $182.6$~[Mbps] and $212.9$~[Mbps] offered by the random and fixed power allocation on average, respectively. Thanks to the advanced scheduling policy, Algorithm~\ref{Algorithm1} provides the sum throughput  $226.5$~[Mbps], which is $24.0\%$ better than the baseline. For a system with $50$ users, the proposed advanced scheduling policy outperforms the random power allocation $26.5\%$.

In Figs.~\ref{Fig:Scheduled} and \ref{Fig:Unscheduled}, we show the percentage of scheduled and unscheduled users, respectively, by exploiting Algorithm~\ref{Algorithm1}. If only $20$ users are available in the coverage area, the satellite-terrestrial cooperative network can serve most of them with $99.2\%$ scheduled users and $0.8\%$  unscheduled users. Nevertheless, a remarkable portion of available users should be ignored if the network density increases. In more detail, with $50$ users in the network, only the portion of $87.3\%$ available users are admitted to the service. In contrast, the remaining of $12.7\%$ available users is out of service. The obtained results demonstrate the roles of user scheduling and power allocation in improving system performance for large network dimensions. 

% \usepackage{multirow}
\begin{table*}[t]
	\centering
	\caption{Performance comparison between learning-based and model-based approaches} \label{Table:Perf}
	\begin{tabular}{|c|ccc|cc|}
		\hline
		\multirow{2}{*}{\begin{tabular}[c]{@{}l@{}}Number\\ of users \end{tabular}} & \multicolumn{3}{c|}{Sum throughput [Mbps]}                             & \multicolumn{2}{c|}{Run-time {[}ms{]}}                                                              \\ \cline{2-6} 
		& \multicolumn{1}{l|}{GNN-Stable} & \multicolumn{1}{l|}{GNN-Robust} & Algorithm~\ref{Algorithm1} & \multicolumn{1}{l|}{\begin{tabular}[c]{@{}l@{}}GNN-Stable,\\ GNN-Robust\end{tabular}} & Algorithm~\ref{Algorithm1} \\ \hline
		$K=20$                                                                         & \multicolumn{1}{l|}{170.5}           & \multicolumn{1}{l|}{171.4}           &         167.1    & \multicolumn{1}{l|}{30}                                                                 &        621     \\ \hline
		$K=30$                                                                         & \multicolumn{1}{l|}{228.6}           & \multicolumn{1}{l|}{228.6}           &     226.5        & \multicolumn{1}{l|}{34}                                                                 &       1667      \\ \hline
		$K=40$                                                                         & \multicolumn{1}{l|}{275.6}           & \multicolumn{1}{l|}{275.9}           &     279.3        & \multicolumn{1}{l|}{35}                                                                 &         2939    \\ \hline
		$K=50$                                                                         & \multicolumn{1}{l|}{317.8}           & \multicolumn{1}{l|}{318.9}           &     321.9        & \multicolumn{1}{l|}{36}                                                                 &        5089     \\ \hline
	\end{tabular}
\end{table*}
\subsection{Model-based versus Learning-based Approaches}
 Table~\ref{Table:Perf} compares the performance of the learning-based and model-based approaches to solve problem~\eqref{Prob:NMSEk} by exploiting the SHGNN in Section~\ref{Sec:DL} and Algorithm~\ref{Algorithm1} in Section~\ref{SubSec:IterAl}, respectively. We consider two different learning-based approaches consisting of the GNN-stable in which the neural network is trained by utilizing measurements from $K=30$ users, while the testing phase is applied for a communication system with a various number of users with $K \in \{20, 30, 40, 50\}$. In contrast, the GNN-Robust is a benchmark with an equal number of users in both of the phases, i.e., the training and testing. Regarding the sum  throughput, the learning-based approaches are very competitive with the model-based approach. Thanks to the unsupervised learning, the GNN-Robust can provide $2.6\%$ better sum  throughput than Algorithm~\ref{Algorithm1}. Besides, GNN-Stable manifests the scalability of our proposed learning-based approach by only training a neural network. It can be applied to a system serving different numbers of users with slightly worse performance than the robust approach, i.e., GNN-Robust. In particular, GNN-Stable only offers $0.5\%$ lower sum throughput than GNN-Robust.  Regarding runtime, the two  learning-based approaches can predict the solution to problem~\eqref{Prob:NMSEk} in the order of milliseconds. Notably, the time consumption only slightly increases $20\%$ as the number of users grows from $K=20$ to $K=50$. In contrast, Algorithm~\ref{Algorithm1} has much higher time consumption than the two previous benchmarks, which scales up $141.4\times$. Furthermore, Algorithm~\ref{Algorithm1} does not scale with the network dimension well. If the number of users increases $2.5\times$ from $K=20$ to $K=50$, the time consumption grows $8.2\times$.

\section{Conclusion} \label{Sec:Conclusion}
In this paper, we have demonstrated the benefits of coherent signal processing from the satellite and ground sectors in enhancing the throughput of complicated networks in which many users simultaneously request to be served in a large coverage area. For such, a novel network architecture was proposed where a satellite and distributed APs are integrated into a unified framework. In addition, users are classified into different subsets and some of them may be ignored from service due to the limited power budget. The uplink ergodic throughput is derived under several practical conditions.  A sum throughput maximization problem is formulated to find the optimal power coefficients and the subset of served users, contained on the transmit power limitation at each user and the availability of channel statistics only. An effective algorithm is proposed to overcome the inherent non-convexity and to attain a low-complexity solution in polynomial time by exploiting the alternating optimization. By utilizing numerical results, we demonstrated that the integrated satellite-terrestrial cell-free massive MIMO IoT systems can provide satisfied services to the vast majority of available users. Still, several users under harsh conditions should be preserved in the queue to maximize the total throughput.   
\appendix
\subsection{Useful Lemmas}
This section presents two useful lemmas utilized for the ergodic throughput analysis. 
\begin{lemma}\cite[Lemma~7]{Chien2020book} \label{Lemma:Supp1}
For two independent random vectors, $\mathbf{u} \sim \mathcal{CN}(\bar{\mathbf{u}}, \mathbf{R}_u)$ and  $\mathbf{v} \sim \mathcal{CN}(\bar{\mathbf{v}}, \mathbf{R}_v)$ with the mean values $\bar{\mathbf{u}}, \bar{\mathbf{v}} \in \mathbb{C}^{N}$ and the covariance matrices $\mathbf{R}_u, \mathbf{R}_v \in \mathbb{C}^{N \times N}$, it holds that
\begin{equation}
\mathbb{E}\{ |\mathbf{u}^H \mathbf{v}|^2 \} = \mathrm{tr}( \mathbf{R}_v \mathbf{R}_u ) + \bar{\mathbf{v}}^H \mathbf{R}_u \bar{\mathbf{v}} +  \bar{\mathbf{u}}^H \mathbf{R}_v \bar{\mathbf{u}} + | \bar{\mathbf{u}}^H \bar{\mathbf{v}}|^2.
\end{equation}
\end{lemma}
\begin{lemma} \label{Lemma:Moment4v1}
For the two correlated random vectors, $\mathbf{u} = \mathbf{R}_u^{1/2} \mathbf{m} + \bar{\mathbf{u}}$ and $\mathbf{v} = \mathbf{R}_v^{1/2} \mathbf{m} + \bar{\mathbf{v}}$ with $\mathbf{m} \sim \mathcal{CN}(\mathbf{0}, \mathbf{I}_N)$, %and a deterministic matrix $\mathbf{M} \in \mathbb{C}^{N \times N}$
it holds that
\begin{align}
& \mathbb{E} \{ \mathbf{u}^H \mathbf{v} \} = \mathrm{tr} \big(  (\mathbf{R}_u^H)^{1/2} \mathbf{R}_v^{1/2} \big) + \bar{\mathbf{u}}^H \bar{\mathbf{v}}, \label{eq:ExpInner}\\ 	
& \mathbb{E} \{ |\mathbf{u}^H \mathbf{v}|^2 \} =  |\bar{\mathbf{v}}^H \bar{\mathbf{u}} |^2 +  2 \mathsf{Re} \left\{ \bar{\mathbf{v}}^H \bar{\mathbf{u}} \mathrm{tr}\big((\mathbf{R}_u^{H})^{1/2} \mathbf{R}_v^{1/2} \big) \right\} \notag \\
& + \bar{\mathbf{u}}^H  \mathbf{R}_v  \bar{\mathbf{u}} +  \bar{\mathbf{v}}^H  \mathbf{R}_u  \bar{\mathbf{v}}  +  \big| \mathrm{tr}\big( (\mathbf{R}_u^H )^{1/2}  \mathbf{R}_v^{1/2} \big) \big|^2 + \mathrm{tr} (  \mathbf{R}_v \mathbf{R}_u).
\end{align}
\end{lemma}
\begin{proof}
By utilizing the structure of $\mathbf{u}$ and $\mathbf{v}$, we compute the expectation of the inner product in \eqref{eq:ExpInner} as follows
\begin{equation}
\begin{split}
& \mathbb{E} \{ \mathbf{u}^H \mathbf{v} \} =  \mathbb{E} \big\{ ( \mathbf{m}^H (\mathbf{R}_u^H)^{1/2} + \bar{\mathbf{u}}^H )(\mathbf{R}_v^{1/2} \mathbf{m} + \bar{\mathbf{v}}) \big\} \\
& =  \mathbb{E} \big\{  \mathbf{m}^H (\mathbf{R}_u^H)^{1/2} \mathbf{R}_v^{1/2} \mathbf{m} \big\} + \bar{\mathbf{u}}^H \bar{\mathbf{v}} \\
&= \mathrm{tr} \big(  (\mathbf{R}_u^H)^{1/2} \mathbf{R}_v^{1/2} \big) + \bar{\mathbf{u}}^H \bar{\mathbf{v}},
\end{split}
\end{equation}
thanks to the zero mean of random vector $\mathbf{m}$ and therefore some expectations have vanished. The second moment of $\mathbf{u}^H \mathbf{v}$ is attained by utilizing the similar steps as in \cite[Lemma~5]{ozdogan2019massive}. 
\end{proof}
\subsection{Proof of Theorem~\ref{Theorem:ClosedForm}} \label{Appendix:ClosedForm}
The main proof is to compute all the expectations in \eqref{eq:SINRk}. By utilizing the definition \eqref{eq:zkkprime} with $k' = k$, we compute the numerator of \eqref{eq:SINRk} as
\begin{equation} \label{eq:Numv1}
\begin{split}
&|\mathbb{E} \{z_{kk}\}|^2 \stackrel{(a)}{=} \left| \mathbb{E} \{ \hat{\mathbf{g}}_k^H \mathbf{g}_{k} \}  +  \sum\nolimits_{m=1}^M  \mathbb{E}\{ \hat{g}_{mk}^\ast  g_{mk} \} \right|^2\\
& \stackrel{(b)}{=} \left| \mathbb{E} \{ \|\hat{\mathbf{g}}_k \|^2 \}  +  \sum\nolimits_{m=1}^M  \mathbb{E}\{ |\hat{g}_{mk}|^2 \} \right|^2 \\
& \stackrel{(c)}{=} \left| \|\bar{\mathbf{g}}_k\|^2 +  p \tau_p \mathrm{tr}(\mathbf{R}_k \pmb{\Phi}_k \mathbf{R}_k)  +  \sum\nolimits_{m=1}^M  \gamma_{mk} \right|^2,
\end{split}
\end{equation}
where $(a)$ is attained by the MRC technique; $(b)$ is because the estimation errors have zero means as shown in Lemma~\ref{Lemma:Est}; and $(c)$ is obtained from the distributions of the channel estimates. We now compute the first part in the denominator of \eqref{eq:SINRk} as
\begin{multline} \label{eq:Denov1}
\sum\nolimits_{k'=1}^K \rho_{k'}  \mathbb{E}\{ |z_{kk'}|^2 \} = \underbrace{\sum\nolimits_{k' \in \mathcal{P}_k } \rho_{k'}  \mathbb{E}\{ |z_{kk'}|^2 \}}_{\triangleq \mathsf{MI}_{1}} \\ 
+   \underbrace{\sum\nolimits_{k' \notin \mathcal{P}_k } \rho_{k'}  \mathbb{E}\{ |z_{kk'}|^2}_{\triangleq \mathsf{MI}_{2}}  \}, 
\end{multline}
which demonstrates the coherent and noncoherent mutual interference, explicitly expressed by the pilot reuse set $\mathcal{P}_k$. For $k' \in \mathcal{P}_k$, we process $\mathbb{E}\{ | z_{kk'} |^2 \}$ by utilizing its definition in \eqref{eq:zkkprime}  as follows
\begin{equation} \label{eq:zkkGain}
\begin{split}
& \mathbb{E} \{ |z_{kk'}|^2 \} = \mathbb{E} \{ | a_{kk'}  + \tilde{a}_{kk'} + b_{kk'} + \tilde{b}_{kk'} |^2 \} \\
&= \mathbb{E} \{ | a_{kk'} |^2 \} +  \mathbb{E} \{ | \tilde{a}_{kk'} |^2 \} + \mathbb{E} \{ | b_{kk'} |^2 \} + \\
& \mathbb{E} \{ | \tilde{b}_{kk'} |^2 \} + \mathbb{E}\{ a_{kk'} b_{kk'}^\ast \}  + \mathbb{E}\{ a_{kk'}^\ast b_{kk'} \},
\end{split}
\end{equation}
where $a_{kk'} = \mathbf{u}_k^H \hat{\mathbf{g}}_{k'}$, $\tilde{a}_{kk'} = \mathbf{u}_k^H \mathbf{e}_{k'}$, $b_{kk'} = \sum_{m=1}^M  u_{mk}^\ast \hat{g}_{mk'}$, and $\tilde{b}_{kk'} = \sum_{m=1}^M u_{mk}^\ast e_{mk'}$. Because of the zero mean of additive noise, the remaining expectations vanish in \eqref{eq:zkkGain}.  By noting that $\mathbf{u}_k = \hat{\mathbf{g}}_k$, $\mathbb{E} \{ | a_{kk'} |^2 \}$ in \eqref{eq:zkkGain} is computed as
\begin{equation} \label{eq:akkprime}
\mathbb{E} \{ | a_{kk'} |^2 \} =  \mathbb{E}\{ | \hat{\mathbf{g}}_k^H \hat{\mathbf{g}}_{k'}|^2  \}.
\end{equation}
Since two users~$k$ and $k'$ share the same pilot signal, we can use the channel estimate structure in \eqref{eq:ChannelEstgk} to represent   $\hat{\mathbf{g}}_k$ and $\hat{\mathbf{g}}_{k'}$ as
\begin{align}
	\hat{\mathbf{g}}_k = \bar{\mathbf{g}}_k + \sqrt{p\tau_p} \mathbf{R}_k \pmb{\Phi}_k^{1/2} \mathbf{m}_k,   \hat{\mathbf{g}}_{k'} = \bar{\mathbf{g}}_{k'} + \sqrt{p\tau_p} \mathbf{R}_{k'} \pmb{\Phi}_k^{1/2} \mathbf{m}_k,
\end{align}
where $\mathbf{m}_k  \sim \mathcal{CN}(\mathbf{0}, \mathbf{I}_N)$. Then, we compute $\mathbb{E}\{ | \hat{\mathbf{g}}_k^H \hat{\mathbf{g}}_{k'}|^2  \}$ in \eqref{eq:akkprime} by utilizing Lemma~\ref{Lemma:Moment4v1}, and then obtain
\begin{equation} \label{eq:hatgkgk}
	\begin{split}
	& \mathbb{E} \{ | a_{kk'} |^2 \} =  |\bar{\mathbf{g}}_{k'}^H \bar{\mathbf{g}}_k |^2 +  2 p \tau_p \mathsf{Re} \left\{ \bar{\mathbf{g}}_{k'}^H \bar{\mathbf{g}}_k \mathrm{tr}\big(\mathbf{R}_{k'}\pmb{\Phi}_k \mathbf{R}_k \big) \right\} +  p \tau_p \bar{\mathbf{g}}_k^H    \\
 & \times \mathbf{R}_{k'} \pmb{\Phi}_k \mathbf{R}_{k'}   \bar{\mathbf{g}}_k   +  p \tau_p  \bar{\mathbf{g}}_{k'}^H \mathbf{R}_k \pmb{\Phi}_k \mathbf{R}_k   \bar{\mathbf{g}}_{k'} + p^2 \tau_p^2  \big| \mathrm{tr}\big( \mathbf{R}_{k'}\pmb{\Phi}_k \mathbf{R}_k \big) \big|^2 + \\
 & p^2 \tau_p^2 \mathrm{tr} (\mathbf{R}_{k'}\pmb{\Phi}_k \mathbf{R}_{k'}  \mathbf{R}_k \pmb{\Phi}_k \mathbf{R}_k)  = \big| \bar{\mathbf{g}}_{k}^H \bar{\mathbf{g}}_{k'} +  p \tau_p \mathrm{tr}(\mathbf{R}_{k'} \pmb{\Phi}_k \mathbf{R}_k)\big|^2  +   \\
 & p \tau_p \bar{\mathbf{g}}_k^H \mathbf{R}_{k'} \pmb{\Phi}_k \mathbf{R}_{k'}  \bar{\mathbf{g}}_k  +  p \tau_p  \bar{\mathbf{g}}_{k'}^H \mathbf{R}_k \pmb{\Phi}_k \mathbf{R}_k   \bar{\mathbf{g}}_{k'} + \\
 & p^2 \tau_p^2 \mathrm{tr} (\mathbf{R}_{k'}\pmb{\Phi}_k \mathbf{R}_{k'}  \mathbf{R}_k \pmb{\Phi}_k \mathbf{R}_k ).
	\end{split}
\end{equation}
By applying Lemma~\ref{Lemma:Supp1}, we can compute $\mathbb{E} \{ | \tilde{a}_{kk'} |^2 \}$ in \eqref{eq:zkkGain} in a closed form as 
\begin{equation} \label{eq:ekhatgk}
	\begin{split}
		&\mathbb{E} \{ | \tilde{a}_{kk'} |^2 \} =  \mathbb{E} \left\{ | \hat{\mathbf{g}}_{k}^H \mathbf{e}_{k'} |^2 \right\} \stackrel{(a)}{=}    \bar{\mathbf{g}}_{k}^H \left( \mathbf{R}_{k'} - p \tau_p \mathbf{R}_{k'} \pmb{\Phi}_{k'} \mathbf{R}_{k'} \right)  \bar{\mathbf{g}}_{k} \\
  & + p \tau_p \mathrm{tr} \left(  \left( \mathbf{R}_{k'} - p \tau_p \mathbf{R}_{k'} \pmb{\Phi}_{k'} \mathbf{R}_{k'} \right) \mathbf{R}_k \pmb{\Phi}_k \mathbf{R}_k  \right) \\
		& \stackrel{(b)}{=}  \bar{\mathbf{g}}_{k}^H \mathbf{R}_{k'} \bar{\mathbf{g}}_{k} -  p \tau_p \bar{\mathbf{g}}_{k}^H  \mathbf{R}_{k'} \pmb{\Phi}_{k} \mathbf{R}_{k'}   \bar{\mathbf{g}}_{k} +  p \tau_p \mathrm{tr}( \mathbf{R}_{k'} \mathbf{R}_k \pmb{\Phi}_k \mathbf{R}_k  ) \\
  & - p^2 \tau_p^2 \mathrm{tr}( \mathbf{R}_{k'} \pmb{\Phi}_k \mathbf{R}_{k'} \mathbf{R}_k \pmb{\Phi}_k \mathbf{R}_k),
	\end{split}
\end{equation}
where $(a)$ is obtained using the covariance matrix in \eqref{eq:Ck} that is expressed for the channel estimation error. From \eqref{eq:Relation}, we have $\hat{g}_{mk'} = c_{mk'} \hat{g}_{mk}/ c_{mk}, \forall k' \in \mathcal{P}_k$. Therefore, $\mathbb{E} \{ | b_{kk'} |^2 \}$ in the last equality of \eqref{eq:zkkGain} is computed as follows
\begin{equation}
	\begin{split}
		&\mathbb{E} \{ | b_{kk'} |^2 \} = \mathbb{E}\left\{ \left| \sum\nolimits_{m=1}^M \frac{c_{mk'} }{c_{mk} }   | \hat{g}_{mk}|^2 \right|^2 \right\} \\
  & = \sum\nolimits_{m=1}^M  \sum\nolimits_{m'=1}^M \frac{c_{mk'}  c_{m'k'} }{c_{mk} c_{m'k}  }   \mathbb{E}\{ | \hat{g}_{m'k}|^2 | \hat{g}_{mk}|^2 \}  \\
		& \stackrel{(a)}{=} 2 \sum_{m=1}^M \frac{c_{mk'}^2 }{c_{mk}^2}   \gamma_{mk}^2  + \sum_{m=1}^M  \sum_{m'=1, m'\neq m}^M \frac{c_{mk'}  c_{m'k'} }{c_{mk} c_{m'k}  }  \gamma_{m'k} \gamma_{mk} \\
		&= \sum_{m=1}^M   \gamma_{mk} \gamma_{mk'} + \left| \sum_{m=1}^M  \frac{c_{mk'}  }{c_{mk}  }  \gamma_{mk} \right|^2,
	\end{split}
\end{equation}
where $(a)$ is attained by applying Lemma~\ref{Lemma:Moment4v1} and the channel estimates attained in Lemma~\ref{Lemma:Est}. Besides, $\mathbb{E} \{ | \tilde{b}_{kk'} |^2 \}$ in the last equality of \eqref{eq:zkkGain} is computed as follows
\begin{equation}
\mathbb{E} \{ | \tilde{b}_{kk'} |^2 \} =  \sum\nolimits_{m=1}^M  \mathbb{E}\{  | \hat{g}_{mk}^\ast  e_{mk'}|^2 \}  =     \sum\nolimits_{m=1}^M    \gamma_{mk} (\beta_{mk'} - \gamma_{mk'}),
\end{equation}
since the estimation error and the channel estimate are mutually independent.  $\mathbb{E}\{ a_{kk'}^\ast b_{kk'} \} $  in the last equality of \eqref{eq:zkkGain} is computed as follows
\begin{equation}
\begin{split}
& \mathbb{E}\{ a_{kk'} b_{kk'}^{\ast} \}  = \mathbb{E}\left\{ \sum\nolimits_{m=1}^M \hat{\mathbf{g}}_k^H \hat{\mathbf{g}}_{k'}  \hat{g}_{mk} \hat{g}_{mk'}^\ast \right\} \\
& \stackrel{(a)}{=} \sum\nolimits_{m=1}^M  \frac{c_{mk'}}{c_{mk}} \mathbb{E}\{ \hat{\mathbf{g}}_k^H \hat{\mathbf{g}}_{k'} \} \mathbb{E}\{ |\hat{g}_{mk}|^2 \} \\
&= \sum\nolimits_{m=1}^M  \frac{c_{mk'}}{c_{mk}} \left( p \tau_p \mathrm{tr} (  \mathbf{R}_{k'} \pmb{\Phi}_k \mathbf{R}_k ) + \bar{\mathbf{g}}_k^H \bar{\mathbf{g}}_{k'} \right) \gamma_{mk},
\end{split}
\end{equation}
where $(a)$ is attained by the independence between the satellite and terrestrial channels along with the channel relation in \eqref{eq:Relation}; and $(b)$ is thanks to \eqref{eq:ExpInner}. Similarly, one attains  $\mathbb{E}\{ a_{kk'}^\ast b_{kk'} \}$ in the last equality of \eqref{eq:zkkGain} as follows
\begin{equation}\label{eq:6thterm}
\mathbb{E}\{ a_{kk'}^\ast b_{kk'} \} = \sum\nolimits_{m=1}^M  \frac{c_{mk'}}{c_{mk}} \left( p \tau_p \mathrm{tr} (  \mathbf{R}_{k} \pmb{\Phi}_k \mathbf{R}_{k'} ) + \bar{\mathbf{g}}_{k'}^H \bar{\mathbf{g}}_{k} \right) \gamma_{mk}.
\end{equation}
Plugging \eqref{eq:hatgkgk}--\eqref{eq:6thterm} into \eqref{eq:zkkGain} and doing some algebra, we attain the closed-form expression of $\mathbb{E}\{ |z_{kk'}|^2\}$, and then that of $\mathsf{MI}_1$ is  as follows
\begin{equation} \label{eq:MI1v1}
\begin{split}
&\mathsf{MI}_1 = \sum\nolimits_{k' \in \mathcal{P}_k} \rho_{k'} \left| \bar{\mathbf{g}}_{k}^H \bar{\mathbf{g}}_{k'} +  \mathrm{tr}(\mathbf{R}_{k'} \pmb{\Phi}_k \mathbf{R}_k) +   \sum\nolimits_{m=1}^M  \frac{c_{mk'}  }{c_{mk}  }  \gamma_{mk} \right|^2 \\
&+ \sum\nolimits_{k' \in \mathcal{P}_k} \rho_{k'}  p \tau_p \bar{\mathbf{g}}_{k'}^H \mathbf{R}_k \pmb{\Phi}_k \mathbf{R}_k   \bar{\mathbf{g}}_{k'} + \sum\nolimits_{k' \in \mathcal{P}_k} \rho_{k'}   \bar{\mathbf{g}}_{k}^H \mathbf{R}_{k'} \bar{\mathbf{g}}_{k} + \\
& \sum\nolimits_{k' \in \mathcal{P}_k} \rho_{k'} p \tau_p \mathrm{tr}( \mathbf{R}_{k'} \mathbf{R}_k \pmb{\Phi}_k \mathbf{R}_k  ) + \sum\nolimits_{k' \in \mathcal{P}_k} \rho_{k'}  \sum\nolimits_{m=1}^M \gamma_{mk} \beta_{mk'}.
\end{split}
\end{equation}
Next, the noncoherent interference $\mathsf{MI}_2$ in \eqref{eq:Denov1} is computed as follows
\begin{equation} \label{eq:MI2v2}
\begin{split}
& \mathsf{MI}_2 =	 \sum\nolimits_{k' \notin \mathcal{P}_k } \rho_{k'}  \mathbb{E}\{ | \hat{\mathbf{g}}_k^H \mathbf{g}_{k'}|^2  \} +   \sum\nolimits_{k' \notin \mathcal{P}_k } \sum\nolimits_{m=1}^M \rho_{k'}   \mathbb{E}\{  | \hat{g}_{mk}^\ast  g_{mk'}|^2 \}\\
& =  p \tau_p \sum\nolimits_{k' \notin \mathcal{P}_k } \rho_{k'}   \mathrm{tr}( \mathbf{R}_{k'}\mathbf{R}_k \pmb{\Phi}_k \mathbf{R}_k )  +  p \tau_p \sum\nolimits_{k' \notin \mathcal{P}_k } \rho_{k'}  \times   \\
& \bar{\mathbf{g}}_{k'}^H \mathbf{R}_k  \pmb{\Phi}_k \mathbf{R}_k   \bar{\mathbf{g}}_{k'} + \sum\nolimits_{k' \notin \mathcal{P}_k } \rho_{k'} \bar{\mathbf{g}}_{k}^H \mathbf{R}_{k'}  \bar{\mathbf{g}}_{k}  +    \sum\nolimits_{k' \notin \mathcal{P}_k } \rho_{k'} |\bar{\mathbf{g}}_{k}^H  \bar{\mathbf{g}}_{k'} |^2\\
& + \sum\nolimits_{k' \notin \mathcal{P}_k } \sum\nolimits_{m=1}^M \rho_{k'}    \gamma_{mk} \beta_{mk'},
\end{split}
\end{equation}
thanks to the mutually independent channels from the users utilizing the orthogonal pilot signals. By using \eqref{eq:MI1v1} and \eqref{eq:MI2v2} into \eqref{eq:Denov1}, we can attain  the first part in the denominator of \eqref{eq:SINRk} in closed form as 
\begin{equation} \label{Eq:FirstDe}
\begin{split}
&\sum\nolimits_{k'=1}^K \rho_{k'}  \mathbb{E}\{ |z_{kk'}|^2 \} =  \sum\nolimits_{k' \in \mathcal{P}_k} \rho_{k'} \Big| \bar{\mathbf{g}}_{k}^H \bar{\mathbf{g}}_{k'} +  \mathrm{tr}(\mathbf{R}_{k'} \pmb{\Phi}_k \mathbf{R}_k) \\
& +   \sum\nolimits_{m=1}^M  \frac{c_{mk'}  }{c_{mk}  }  \gamma_{mk} \Big|^2 +  p \tau_p \sum\nolimits_{k \in \mathcal{Q}} \rho_{k'}   \bar{\mathbf{g}}_{k'}^H \mathbf{R}_k \pmb{\Phi}_k \mathbf{R}_k   \bar{\mathbf{g}}_{k'} + \\
& \sum\nolimits_{k \in \mathcal{Q}} \rho_{k'}   \bar{\mathbf{g}}_{k}^H \mathbf{R}_{k'} \bar{\mathbf{g}}_{k}   + p \tau_p \sum\nolimits_{k \in \mathcal{Q}} \rho_{k'}  \mathrm{tr}( \mathbf{R}_{k'} \mathbf{R}_k \pmb{\Phi}_k \mathbf{R}_k  ) \\
& +  \sum\nolimits_{k' \notin \mathcal{P}_k } \rho_{k'} |\bar{\mathbf{g}}_{k}^H  \bar{\mathbf{g}}_{k'} |^2 + \sum\nolimits_{k \in \mathcal{Q}} \rho_{k'}  \sum\nolimits_{m=1}^M \gamma_{mk} \beta_{mk'}.
\end{split}
\end{equation}
Next, the noise power from the satellite link is computed in the closed-form expression as
\begin{equation} \label{eq:3De}
\begin{split}
& \mathbb{E} \big\{ \big| \mathbf{u}_k^H \mathbf{w} \big|^2 \big\} \stackrel{(a)}{=}  \mathbb{E} \big\{ \big| \mathbf{u}_k^H \mathbb{E}\{ \mathbf{w} \mathbf{w}^H \} \mathbf{u}_k \big|^2 \big\} \\
&= \sigma^2 \|\bar{\mathbf{g}}_k\|^2 +  p \tau_p  \sigma^2 \mathrm{tr}(\mathbf{R}_k \pmb{\Phi}_k \mathbf{R}_k), 
\end{split}
\end{equation}
where $(a)$ is attained by the independence of the channel estimate and noise. In a similar manner, the noise power from all the $M$ APs is driven in closed form as
\begin{equation}\label{eq:4De}
\sum\nolimits_{m=1}^M  \mathbb{E} \big\{ | u_{mk}^\ast w_m |^2 \big\} = \sigma_a^2 \sum\nolimits_{m=1}^M  \mathbb{E} \big\{ | \hat{g}_{mk} |^2 \big\} =  \sigma_a^2 \sum\nolimits_{m=1}^M  \gamma_{mk}.
\end{equation}
By plugging \eqref{eq:Numv1}, \eqref{Eq:FirstDe}, \eqref{eq:3De}, and \eqref{eq:4De} into \eqref{eq:SINRk}, we obtain the closed-form expression as in \eqref{eq:ClosedSINR} after doing some algebra.
\subsection{Proof of Theorem~\ref{Theorem:ClosedFormSol}}\label{Appendix:ClosedFormSol}
The iteration indices are ignored for the sake of simplicity. By computing  the first derivative of the Lagrangian function in \eqref{Prob:NMSEkE1} with respect to $v_k$, we obtain
\begin{multline}
\frac{\partial \mathcal{L} }{\partial v_k} = 2\left(\tilde{\rho}_k v_k \left(\|\bar{\mathbf{g}}_k\|^2 +  p \tau_p \mathrm{tr}(\mathbf{R}_k \pmb{\Phi}_k \mathbf{R}_k )  +  \sum_{m=1}^M \gamma_{mk} \right) -1 \right)  \\
\times\left( \|\bar{\mathbf{g}}_k\|^2 +  p \tau_p \mathrm{tr}(\mathbf{R}_k \pmb{\Phi}_k \mathbf{R}_k )  +  \sum\nolimits_{m=1}^M \gamma_{mk} \right)\tilde{\rho}_k + 2v_k  \delta_k , 
\end{multline}
and  the optimal solution to $v_k$ is achieved by solving the equation $\partial \mathcal{L}/\partial v_k = 0$ as shown in the theorem. We derive the first derivative of \eqref{Prob:NMSEkE1} with respect to $\alpha_k$ as
\begin{equation}
\frac{\partial \mathcal{L} }{\partial \alpha_k}  = e_k - \frac{1}{\alpha_k},
\end{equation}
and the optimal solution to $\alpha_k$ is obtained as in the theorem by solving the equation $\partial \mathcal{L} /\partial \alpha_k = 0$. Based on the mutual interference formulated in \eqref{eq:MIk},  let us consider user~$k''$ and define
\begin{multline}
 \mathsf{CI}_{k''} = \sum\nolimits_{k \in \mathcal{P}_{k''} \setminus \{k''\}} \rho_{k'} \Big| \bar{\mathbf{g}}_{k''}^H \bar{\mathbf{g}}_{k} + p \tau_p   \mathrm{tr}(\mathbf{R}_{k} \pmb{\Phi}_{k''} \mathbf{R}_{k''}) \\
 +   \sum\nolimits_{m=1}^M  \frac{c_{mk}  }{c_{mk''}  }  \gamma_{mk''} \Big|^2,
\end{multline}
 then its first derivative with respect to $\tilde{\rho}_k$ is computed as follows
\fontsize{8}{8}{\begin{align} \label{eq:CIkk}
& \frac{\partial \mathsf{CI}_{k''} }{\partial \tilde{\rho}_k} = \notag \\
& \begin{cases}
		2 \tilde{\rho}_{k} \left| \bar{\mathbf{g}}_{k''}^H \bar{\mathbf{g}}_{k} + p \tau_p   \mathrm{tr}(\mathbf{R}_{k} \pmb{\Phi}_{k''} \mathbf{R}_{k''}) +   \sum\limits_{m=1}^M  \frac{c_{mk}  }{c_{mk''}  }  \gamma_{mk''} \right|^2, &  k \in \mathcal{P}_{k''} \setminus \{k''\},\\
		0, & \mbox{otherwise},
\end{cases}  
\end{align}}
which relies on the pilot reuse pattern in \eqref{eq:PilotPattern}. Let us define $\mathsf{NI}_{k''} =  \mathsf{MI}_{k''} + \mathsf{NO}_{k''} -  \mathsf{CI}_{k''} $, then its first derivative with respect to $\tilde{\rho}_k$ is computed as follows
\begin{align}
\frac{\partial \mathsf{NI}_{k''} }{\partial \tilde{\rho}_k} = \begin{cases} \chi_{kk''} ,& \mbox{if } k \in \mathcal{P}_{k''},\\
2\tilde{\rho}_{k} |\bar{\mathbf{g}}_{k''}^H  \bar{\mathbf{g}}_{k} |^2 + \chi_{kk''},& \mbox{otherwise},
\end{cases} 
\end{align}
where the following definition of $\chi_{kk''}$ holds
\begin{multline} \label{eq:xikk}
\chi_{kk''} = 2 \tilde{\rho}_{k}  p \tau_p \bar{\mathbf{g}}_{k}^H \mathbf{R}_{k''} \pmb{\Phi}_{k''} \mathbf{R}_{k''}   \bar{\mathbf{g}}_{k} + 2 \tilde{\rho}_{k} \bar{\mathbf{g}}_{k''}^H \mathbf{R}_{k} \bar{\mathbf{g}}_{k''} \\
+ 2 \tilde{\rho}_{k}  p \tau_p   \mathrm{tr}( \mathbf{R}_{k} \mathbf{R}_{k''} \pmb{\Phi}_{k''} \mathbf{R}_{k''}  ) +  2 \sum\limits_{m=1}^M \tilde{\rho}_{k}  \gamma_{mk''} \beta_{mk}. 
\end{multline}
Consequently, we can derive the first-order derivative of the Lagrangian function in \eqref{Prob:NMSEkE1} with respect to $\tilde{\rho}_k$ as in \eqref{eq:LagrangianPv1}.
\begin{figure*}
\begin{equation} \label{eq:LagrangianPv1}
\begin{split}
&\frac{\partial \mathcal{L} }{\partial \tilde{\rho}_k} =  2 \alpha_k \left(\tilde{\rho}_k v_k \left( \|\bar{\mathbf{g}}_k\|^2 +  p \tau_p \mathrm{tr}(\mathbf{R}_k \pmb{\Phi}_k \mathbf{R}_k )  +  \sum\limits_{m=1}^M \gamma_{mk} \right) -1  \right) \left( \|\bar{\mathbf{g}}_k\|^2 +  p \tau_p \mathrm{tr}(\mathbf{R}_k \pmb{\Phi}_k \mathbf{R}_k )  +  \sum\limits_{m=1}^M \gamma_{mk} \right) v_k  +  \sum_{k'' \in \mathcal{K}} \alpha_{k''} v_{k''}^2  \frac{ \partial \delta_{k''} }{\partial \tilde{\rho}_k }  \\
&+ 2 \lambda_k \tilde{\rho}_k \stackrel{(a)}{=}  2 \alpha_k \left(\tilde{\rho}_k v_k \left( \|\bar{\mathbf{g}}_k\|^2 +  p \tau_p \mathrm{tr}(\mathbf{R}_k \pmb{\Phi}_k \mathbf{R}_k )  +  \sum\limits_{m=1}^M \gamma_{mk} \right)-1 \right)\left( \|\bar{\mathbf{g}}_k\|^2 +  p \tau_p \mathrm{tr}(\mathbf{R}_k \pmb{\Phi}_k \mathbf{R}_k)  +  \sum\limits_{m=1}^M \gamma_{mk} \right) v_k + 2 \lambda_k \tilde{\rho}_k + \\
& \sum_{k'' \in \mathcal{K}} \alpha_{k''} v_{k''}^2  \frac{ \partial \mathsf{CI}_{k''} }{\partial \tilde{\rho}_k } + \sum_{k''\in \mathcal{K}} \alpha_{k''} v_{k''}^2  \frac{ \partial \mathsf{NI}_{k''} }{\partial \tilde{\rho}_k } 
\end{split}
\end{equation}
\vspace{-1cm}
\end{figure*}
where $(a)$ is attained by utilizing $\delta_k$ on its definition. Plugging \eqref{eq:CIkk}--\eqref{eq:xikk} into \eqref{eq:LagrangianPv1}, we obtain the result in \eqref{eq:L}.
\begin{figure*}
\begin{equation} \label{eq:L}
\begin{split}
 & \frac{\partial \mathcal{L} }{\partial \tilde{\rho}_k} =2 \alpha_k \left(\tilde{\rho}_k v_k \left( \|\bar{\mathbf{g}}_k\|^2 +  p \tau_p \mathrm{tr}(\mathbf{R}_k \pmb{\Phi}_k \mathbf{R}_k )  +  \sum\limits_{m=1}^M \gamma_{mk} \right) -1 \right) \left( \|\bar{\mathbf{g}}_k\|^2 +  p \tau_p \mathrm{tr}(\mathbf{R}_k \pmb{\Phi}_k \mathbf{R}_k )  +  \sum\limits_{m=1}^M \gamma_{mk} \right) v_k + 2 \lambda_k \tilde{\rho}_k  \\
& +  \sum_{k'' \in \mathcal{K} }  \alpha_{k''} v_{k''}^2 \chi_{kk''} +  2\sum_{k'' \notin \mathcal{P}_k }  \alpha_{k''} v_{k''}^2  \tilde{\rho}_k |\bar{\mathbf{g}}_{k''}^H \bar{\mathbf{g}}_k|^2 +  2 \sum_{k'' \in \mathcal{P}_{k} \setminus \{k\} } \alpha_{k''} v_{k''}^2  \tilde{\rho}_{k} \left| \bar{\mathbf{g}}_{k''}^H \bar{\mathbf{g}}_{k} + p \tau_p   \mathrm{tr}(\mathbf{R}_{k} \pmb{\Phi}_{k''} \mathbf{R}_{k''}) +   \sum\limits_{m=1}^M  \frac{c_{mk}  }{c_{mk''}  }  \gamma_{mk''} \right|^2 
\end{split} 
\end{equation}
\hrule
\vspace{-0.5cm}
\end{figure*}
with noting that $k'' \in \mathcal{P}_{k}$ is equivalent to $k \in \mathcal{P}_{k''}$. The optimal solution to $\tilde{\rho}_k$ is attained by solving the equation $\partial \mathcal{L} /\partial \tilde{\rho}_k = 0$ as follows
\begin{equation} \label{eq:rhok}
\tilde{\rho}_k =  \alpha_k  v_k \left( \|\bar{\mathbf{g}}_k\|^2 +  p \tau_p \mathrm{tr}(\mathbf{R}_k \pmb{\Phi}_k \mathbf{R}_k)  +  \sum\limits_{m=1}^M \gamma_{mk} \right)  / (t_k + \lambda_k),
\end{equation}
where $t_k$ is defined in \eqref{eq:tkn}, but here without the iteration index. Moreover, the optimal solution expressed in \eqref{eq:rhok} should satisfy the complementary slackness condition, which is
\begin{equation} \label{eq:ComSlack}
\lambda_k (\tilde{\rho}_k^2 - P_{\max,k}) = 0.
\end{equation}
The solution to $\tilde{\rho}_k$ is further defined by combining \eqref{eq:rhok} and \eqref{eq:ComSlack} as follows
\begin{equation}
\tilde{\rho}_k = \begin{cases}
\min(\bar{\rho}_k, \sqrt{P_{\max,k}}), & \mbox{if } \lambda_k =0,\\
\sqrt{P_{\max,k}}, & \mbox{if } \lambda_k =0,
\end{cases}
\end{equation}
where $\bar{\rho}_k$ is given in \eqref{eq:tkn} that is obtained from \eqref{eq:rhok} by setting $\lambda_k = 0$. We stress that the iterative mechanism in Theorem~\ref{Theorem:ClosedFormSol} will converge to a fixed point after a limited number of iterations thanks to a compact feasible domain. We can borrow the main steps in \cite[Theorem~4]{van2018large} to prove that this fixed point is a stationary solution of problem~\eqref{Prob:NMSEkEv2}.
\bibliographystyle{IEEEtran}
\bibliography{IEEEabrv,refs}
\end{document}